\documentclass{IEEEtaes}

\usepackage{color, array,amsthm,newtxmath}
\usepackage{graphicx}
\usepackage{cite}
\usepackage{multirow}
\usepackage{makecell}

\jvol{XX}
\jnum{XX}
\jmonth{XXXXX}
\paper{1234567}
\pubyear{2023}
\doiinfo{TAES.2023.Doi Number}

\newtheorem{theorem}{Theorem}
\newtheorem{lemma}{Lemma}
\newtheorem{proposition}{Proposition}
\begin{document}

\title{Adaptive Smooth Control via Nonsingular Fast Terminal Sliding Modes for Distributed Space Telescope Demonstration Mission by CubeSat Formation Flying}

\author{Soobin Jeon}
\affil{Ph.D Candidate, Department of Astronomy, Yonsei University, Seoul 03722, Republic of Korea} 

\author{Hancheol Cho}
\affil{Associate Professor, Department of Astronomy, Yonsei University, Seoul 03722, Republic of Korea

       Visiting Assistant Professor, Department of Aerospace Engineering, Embry-Riddle Aeronautical University, Daytona Beach, FL 32114, USA} 

\author{Sang-Young Park}
\affil{Professor, Department of Astronomy, Yonsei University, Seoul 03722, Republic of Korea}


\receiveddate{Manuscript received XXXXX 00, 0000; revised XXXXX 00, 0000; accepted XXXXX 00, 0000.\\
}

\corresp{
{\itshape Corresponding author: H. Cho}}

\authoraddress{
Soobin Jeon is with the Department of Astronomy, Yonsei University, Seoul 03722, Republic of Korea 
(e-mail: \href{mailto:forestine@yonsei.ac.kr}{forestine@yonsei.ac.kr}). 
Hancheol Cho is with the Department of Astronomy, Yonsei University, Seoul 03722, Republic of Korea (e-mail: \href{mailto:hancho37@yonsei.ac.kr}{hancho37@yonsei.ac.kr}) and the Department of Aerospace Engineering, Embry-Riddle Aeronautical University, Daytona Beach, FL 32114, USA (e-mail: \href{mailto:choh15@erau.edu}{choh15@erau.edu}). 
Sang-Young Park is with the Department of Astronomy, Yonsei University, Seoul 03722, Republic of Korea (e-mail: \href{mailto:spark@yonsei.ac.kr}{spark624@yonsei.ac.kr}).
}


\markboth{S. JEON ET AL.}{ASC via NFTSM for DST Demonstration by CubeSat Formation Flying}
\maketitle

\begin{abstract}
This paper presents a nonsingular fast terminal sliding mode-based adaptive smooth control methodology for a distributed space telescope demonstration mission. The distributed space telescope has a flexible focal length that corresponds to the relative position in the formation flying concept. The limited specification of a CubeSat generally restricts the performance of actuators – most critically the degrees of freedom of controlled motion. This investigation leads to the development of an adaptive smooth control methodology via nonsingular fast terminal sliding modes. \textcolor{black}{The adaptive smooth control algorithm that was developed for a single-input single-output system is adopted and extended to the relative orbit and attitude control systems of the distributed space telescope. The software simulation is conducted under a real mission, which means the real CubeSat structures, hardware specifications, and operational constraints. The proposed algorithm possesses only seven parameters that can be easily adjusted considering their physical meanings. Furthermore, the pre-designated error bounds are analytically derived, which enhances the applicability of the algorithm to real missions.} The simulation compares the efficiency of the adaptive smooth nonsingular fast terminal sliding mode controller with the linear quadratic regulator \textcolor{black}{and proportional derivative} algorithm\textcolor{black}{s}. The results verify that the adaptive smooth nonsingular fast terminal sliding mode control algorithm shows better control performance in the perspective of the alignment time and the fuel consumption for the distributed space telescope demonstration mission.
\end{abstract}


\section*{Nomenclature}
\addcontentsline{toc}{section}{Nomenclature}
\begin{IEEEdescription}[\IEEEusemathlabelsep\IEEEsetlabelwidth{$\boldsymbol{\omega} = \begin{bmatrix} \omega_{1} & \omega_{2} & \omega_{3} \end{bmatrix}^{T}$}]
\item[$\textbf{\textit{r}}=\begin{bmatrix} x & y & z \end{bmatrix}^{T}$] Relative position vector in the LVLH frame
\item[$\dot{\textbf{\textit{r}}}=\begin{bmatrix} \dot{x} & \dot{y} & \dot{z} \end{bmatrix}^{T}$] Relative velocity vector in the LVLH frame
\item[$\textbf{\textit{q}} = \begin{bmatrix} \textbf{\textit{q}}_{v}^{T} & q_{4} \end{bmatrix}^{T}$] Quaternion
\item[$\textbf{\textit{q}}_{v} = \begin{bmatrix} q_{1} & q_{2} & q_{3} \end{bmatrix}^{T}$] Quaternion vector component
\item[$\boldsymbol{\omega} = \begin{bmatrix} \omega_{1} & \omega_{2} & \omega_{3} \end{bmatrix}^{T}$] Angular velocity
\item[$\textbf{\textit{d}}$] Disturbance 
\item[$d$] Lower bound of disturbance
\item[$D$] Upper bound of disturbance
\item[$m$] Mass of the deputy satellite
\item[$\textbf{R}$] Diagonal matrix comprising position errors
\item[$R_{ei}$] Position error bound of NFTSM controller
\item[$Q_{ei}$] Quaternion error bound of NFTSM controller
\item[$\textbf{Q}_{LQR}$] Q weight matrix of LQR controller
\item[$\textbf{R}_{LQR}$] R weight matrix of LQR controller
\item[$(\cdot)^{i}$] A variable or parameter in an inertial frame
\item[$(\cdot)^{l}$] A variable or parameter in the LVLH frame
\item[$(\cdot)^{b}$] A variable or parameter in the body frame
\item[$(\cdot)_{c}$] A variable or parameter of the chief satellite
\item[$(\cdot)_{d}$] A variable or parameter of the deputy satellite
\item[$(\cdot)_{p}$] A primary vector of TRIAD algorithm
\item[$(\cdot)_{s}$] A secondary vector of TRIAD algorithm
\item[$(\cdot)_{w}$] A variable or parameter of the reaction wheels
\item[$(\cdot)_{orb}$] A variable or parameter of the relative orbit control system
\item[$(\cdot)_{att}$] A variable or parameter of the attitude control system
\item[$(\cdot)_{x,y,z}$] The first, second, and third components of the orbital variable or parameter
\item[$(\cdot)_{1,2,3}$] The first, second, and third components of the attitude variable or parameter
\item[$(\cdot)_{r}$] Reference, desired, or target variable or parameter
\item[$(\cdot)_{e}$] Error variable or parameter

\end{IEEEdescription}

\section{Introduction}
\label{s:1intro}


The advancement of satellite technology has led to the concept of a distributed space system, which aims to replace medium or large single satellites with a number of small satellites\cite{DSS1,DSS2,DSS3}. The distributed space system integrates several satellites into a single architecture and encompasses satellite constellations, formation flying, and distributed space telescopes, etc. Distributed space systems offer the capability that strategically organizes the configuration of satellites in accordance with given mission objectives. Hence, it contributes to expanding the diversity and flexibility of scientific missions utilizing multiple satellites.

Conventional formation flight missions have reference relative trajectories that are designed based on the relative dynamics exhibited by two similar satellites. Namely, the reference trajectory is deliberately designed to retain its stability in the absence of a control input. The Gravity Recovery and Climate Experiment (GRACE) and the Gravity Recovery and Climate Experiment Follow-on (GRACE-FO) mission was a scientific mission that quantifies the gravitational field by analyzing alterations in relative motion\cite{GRACE1,GRACE2}. TerraSAR-X/TanDEM-X (TSX/TDX) satellites acquired ground observation data regularly by applying instantaneous thrust to maintain relative motion within the control window\cite{TSXTDX}. In \cite{CanX45,CanX45_2,CPOD} the formation flight control technology was verified by either modifying the relative orbit for a rendezvous (CanX-4/5) or executing inter-satellite docking with proximity operations (CPOD).


A distributed space telescope is a kind of distributed space system that observes inertial objects by formation flying of dual-spacecraft. The distributed space telescope that is also known as the binary-distributed telescope \cite{AWKoenig} or the virtual telescope \cite{PCCalhoun} refers to a group of multiple satellites equipped with optics and detection payload functions as a single instrument\cite{AWKoenig}. Distributed telescopes have a variable and extensive local length compared to monolithic telescopes and consist of a number of more compact structures, which possesses the advantage of relatively low launch and development costs \cite{AirForce}. This paper uses the term distributed space telescope. 

The Inertial Alignment Hold (IAH) is the technology that maintains a specific configuration toward an inertial target in order to implement distributed space telescopes. The relative distance of the formation flying satellites corresponds to the telescope's baseline or focal length and is designed to range from 1,000 to 72,000 kilometers (astrophysical targets)\cite{MASSIM,NWO} to 50 to 500 meters (heliophysical targets and X-ray images)\cite{ASPIICS,BRDennis}, depending on the observation target and wavelength band. In contrast to conventional formation flight missions, the reference trajectories of the inertial alignment hold missions should be designed as the orbits with the separation of focal lengths with respect to the inertial targets. Since the control law must retain the alignment precisely within the requirements in an unknown space environment, this paper suggests a robust control approach that compensates for any effects of disturbances using continuous thrust.


The robust control system should be able to resist these effects that are neglected in the \textcolor{black}{system design such as disturbances, measurement noise, and unmodeled dynamics} while successfully performing the tasks \cite{SMC1}. \textcolor{black}{Spacecraft disturbances can be divided into external and internal effects \cite{Dist1,Dist2}. The external disturbances arise from the spacecraft’s interaction with the space environment and tend to have low frequency on the order of the orbital rate. These environmental disturbance forces and torques in this paper are summarized in Section V.A. Common sources of orbital and attitude disturbances include Earth’s nonuniform gravity, atmospheric drag, and solar radiation pressure. Orbital dynamics are further influenced by gravitational forces from third bodies, such as the Sun and Moon, while attitude control has additional disturbance torques due to residual magnetic dipole moments. The internal disturbances include internal vibrations which break down to the spacecraft subsystem level such as attitude control, propulsion, thermal and power. In this paper, the uncertainties in mass and moment of inertia are assumed as the internal disturbances and expressed in Eqs. \eqref{eq:uncertainMass},\eqref{eq:uncertainMoI}}. Several robust control techniques have been proposed and have proven their effectiveness in a variety of practical control problems \cite{Robust1,Robust2,Robust3,Robust4}. One common solution for achieving robustness is sliding mode control (SMC) \cite{SMC2,SMC3,SMC4}. The sliding mode control approach allows the system's states to defy the natural motion and follow the designed trajectory, the so-called `sliding surface.'

The feasibility of sliding mode control in formation flying missions has been successfully shown in literature over several decades. The adaptive linear and nonsingular terminal sliding mode control approach achieved submillimeter formation keeping precision even when the operational faults were detected \cite{SMCSFF6}. The continuous and high-order sliding mode controller with consecutive pulse-width modulation was applied to the satellite formation flying in \cite{SMCSFF1}. The sliding mode approach in this application showed high tracking accuracy in the presence of bounded and matched disturbances and uncertainties. The fuzzy-SMC-based technique successfully conducted reconfigurations within a short time span with a minimum control effort for the multiple satellite formation flying problems \cite{SMCSFF5}. 

\textcolor{black}{
The adaptive sliding mode control has been utilized to address the extended concept of satellite formation such as formation antenna array \cite{SMCSFFR1}. The dual quaternion model is first introduced in \cite{SMCSFFR2}. The dual quaternion parametrizes the translational and rotational motion of a rigid body as a 6-degrees-of-freedom equations of motion. Hence, the adaptive continuous sliding mode control in \cite{SMCSFFR1} handles the coupled perturbed system successfully compared to the PD controller and classical SMC. In \cite{SMCSFFR3}, the adaptive fast terminal sliding mode control based on dual quaternion is developed to describe the coupled relative motion of spacecraft formation flying. However, the asymmetric Earth gravitational effect is not considered and the uncertainty bound is assumed to be known. Another adpative terminal sliding mode control designed based on the dual quaternion dynamics deals with the actuator faults in the spacecraft formation flying system \cite{SMCSFFR4}. The control law in \cite{SMCSFFR4} contains the signum function, which produces the discontinuous control input. The pose tracking control for spacecraft proximity operation includes the simultaneous relative position tracking and boresight pointing adjustment \cite{SMCSFFR5}. The prescribed performance control method can accomplish the proximity operations in a desired time, achieve compliance with spatial motion constraints and avoid singularity of the attitude extraction algorithm.
}



The terminal sliding mode (TSM) technique was first proposed for finite-time control for rigid robotic manipulators \cite{TSMrobo,NTSMrobo}. It successfully achieved finite-time stability, although the singularity problem existed \cite{TSMrobo}. The investigation in \cite{NTSMrobo} suggested a modified type of terminal sliding variable by defining a nonsingular terminal sliding mode (NTSM) variable. In the aerospace field, the TSM approach addressed the attitude tracking problem of spacecraft \cite{TSMaero}. However, it had the singularity problem and the research in \cite{NTSMaero} dealt with it in the same manner as in \cite{NTSMrobo}.

The finite-time control (FTC) approach \cite{FTCGZeng} combines the advantages of the linear and terminal sliding mode control approaches. This FTC algorithm as a TSM variant guarantees the convergence of tracking errors in finite time for the relative translational motion problems such as formation-keeping and reconfiguration of electromagnetic formation flying. The sliding variable proposed in \cite{FTCGZeng} is similar to the one that this paper suggests for the relative orbit control system (ROCS). It is proven that the sliding variable in \cite{FTCGZeng} shows faster convergence speed near and far from the equilibrium points than the linear and the conventional terminal sliding modes in \cite{FTCGZeng}. However, the controller contains the term $\left\lvert \textbf{\textit{e}} \right\lvert^{\gamma - 1}$ where $\textbf{\textit{e}}$ is the relative position error and $\gamma \in (0,1)$ is a constant, provoking the singularity problem that the previous research suffered. In addition, the constant gain hardly overcomes the unknown bounds of disturbances and uncertainties. This investigation leads to the proposition of a new nonsingular sliding variable for an adaptive smooth controller.

The study presented in \cite{PMTiwari} designed a nonsingular fast terminal sliding mode (NFTSM), another modified TSM technique. The NFTSM in \cite{PMTiwari} deals with a singular hyper-plane term by increasing the exponent range to $(1,2)$. Besides, the combination of the equivalent control and reaching phase control with an adaptive gain alleviates chattering and achieves fast convergence. The sliding variable for the attitude control system (ACS) in this paper defines the hyper-plane term for quaternion errors in contrast to the one in \cite{PMTiwari} in which both the quaternion and angular velocity errors were applied. Despite this minor difference, the sliding variables in both papers include the hyper-plane and continuous error, and enhance the convergence rate from an arbitrary initial state. In addition, those two sliding variables utilize the increased exponent range that guarantees nonsingularity. However, the adaptive laws in \cite{PMTiwari} might suffer from unintended chattering and need to compromise between the chattering and speed. The adaptive rule in the current study successfully sustains smooth control without any chattering. 

\textcolor{black}{
A lot of literatures address a variety of control plants with sliding mode control. The surface vehicle is one of the control plants that are commonly dealt with. The target and trajectory tracking problems of unmanned surface vehicle has been solved by introducing the barrier function-based adaptive pseudo-inverse controller and nonsingular terminal sliding mode controller with observers \cite{HHe}. The nonsingular terminal sliding mode guarantees the finite time stability to the trajectory tracking system and achieves higher accuracy compared to direct position measurement and adaptive position estimator. On the other hand, the tracking problem of the asymmetric underactuated surface vehicle has been dealt with a finite-time unknown observer-based interactive trajectory tracking control, of which the finite time observer contains the signum function \cite{NWang}. It achieves higher accuracy than interactive trajectory tracking control, finite-time unknown observer cascade-backstepping control, and cascade-backstepping control. }



\textcolor{black}{This paper aims to} achieve precise control of the relative orbit and attitude motion of multiple satellites without a priori knowledge of the system uncertainties or external disturbances. To this end, a new \textcolor{black}{nonsingular fast terminal sliding mode-based adaptive smooth control law is proposed and applied to the CubeSat formation flying mission to demonstrate a distributed space telescope}. \textcolor{black}{The novel} nonsingular fast terminal sliding mode variable is first introduced while overcoming the singularity and achieving a fast convergence to the equilibrium point from any initial state. 

Next, the adaptive smooth control law \textcolor{black}{in \cite{HCho2020} that was originally designed for a single-input single-output (SISO) system is adopted and extended to multi-input multi-output (MIMO) systems, for instance, a formation flying system.} \textcolor{black}{The original algorithm has three key properties: adaptive, smooth, and a smaller number of control parameters. Adaptive control is generally known to have better control performances than the one with a constant gain under disturbances and uncertainties. The absence of the signum function suspends chattering which is a chronic phenomenon for sliding modes. Besides, the algorithm is designed to have only a few parameters which makes a parameter tuning even easier and simpler. The proposed algorithm possesses three key properties from the original one and applied to a formation flying system while dealing with the complexity of the dynamics.}

\textcolor{black}{The efficiency of the adaptive smooth nonsingular fast terminal sliding mode control algorithm proposed in this paper is verified by applying it to a real mission. The CANYVAL (CubeSat Astronomy by NASA and Yonsei using the Virtual ALignment)-C mission aimed to demonstrate the distributed space telescope by CubeSat formation flying and was launched in 2021. Therefore, the structures, configurations, specifications, and operational constraints are all elaborately designed for a real mission. The whole system design can be referred to the reference \cite{CANYVAL-C}. The technology readiness levels (TRL) are a metric developed by the National Aeronautics and Space Administration (NASA) to assess the maturity of a particular technology\cite{TRL1}. The TRL scales in four stages: basic research (TRL 1-3), development (TRL 4-5), demonstration (TRL 6-7), and early development (TRL 8-9)\cite{TRL2}. Universities usually focus on TRL 1-4, which are basic principles observed, technology concept and/or application formulated, analytical and experimental critical function and/or characteristic proof of concept, and component and/or breadboard validation in laboratory environment\cite{TRL1},\cite{TRL3}. The proposed algorithm is built on the theoretical concept for SISO system, developed for the formation flying system, which is a more realistic control plant, and verified the control performance in software simulation. Therefore, the research in this paper enhances the TRL.}

\textcolor{black}{
The control algorithm consists of three sliding variable parameters and four control parameters. The parameters are designed to have physical meanings, so they are easily adjustable. The sliding variable parameters are relevant to the convergence speed of the sliding variable and the weights of the error states relative to the time derivatives. The control parameters determine the pre-designated error bounds and the threshold that the control gain increases or decreases.}

\textcolor{black}{
When the sliding variable is bounded by a small region, the error state is also bounded within a small value. In this investigation, the error bound according to the sliding variable’s bound is analytically derived and expressed as a function of control parameters. The mission requirements induce the desired error bounds, and the error bounds derive the simulation parameters. The simulation result verifies that the proposed algorithm successfully achieves the desired error bounds. It has proven to be beneficial for application in real missions. 
}

\textcolor{black}{The contributions are summarized as below:
\begin{itemize} 
\item A novel type of nonsingular fast terminal sliding mode is developed and applied for ROCS and ACS. 
\item The adaptive smooth control algorithm originally developed for SISO nonlinear systems is adopted and extended for an example of MIMO nonlinear systems, a formation flying system.
\item The efficiency of the proposed algorithm is verified by applying it to a real mission, incorporating real CubeSat structures, actuator specifications, and operational constraints, enhancing the technology readiness level of the theoretical algorithm.
\item The algorithm utilizes only three sliding variable parameters and four control parameters that can be designed to be physically meaningful and easily adjustable.
\item The proposed algorithm can achieve pre-designated error bounds, facilitating its application to real missions. 
\end{itemize}
}


This paper is structured as follows. The first section describes the historical background of the distributed space telescope (DST) and NFTSM, and the aim and the contributions of the current study sequentially. Section~\ref{s:2mission} describes the concept and the characteristics of the Guidance, Navigation, and Control (GN\&C) subsystem of the DST demonstration (CANYVAL-C) mission by CubeSats. The dynamic and kinematic equations of the nominal and error states are explained in Section~\ref{s:3FFS} for ACS and ROCS. Section~\ref{s:4ctrlSys} introduces the new NFTSM and adaptive smooth control law, and proves their stability based on the Lyapunov theory. The control performances described in Section~\ref{s:5sim} verify the feasibility and effectiveness of the adaptive smooth nonsingular fast terminal sliding mode controller for the CubeSat formation flying mission developed in this paper. Finally, Section~\ref{s:6conclusion} concludes the paper.

\section{Mission Description}
\label{s:2mission}

The CANYVAL project is a distributed space telescope demonstration mission by formation flying of 1U and 2U CubeSats and was developed by Yonsei University and National Aeronautics and Space Administration (NASA). The 1U satellite is a DSC (Detector Spacecraft) equipped with a visible wavelength camera, and an OSC (Optics Spacecraft), 2U satellite, is equipped with an occulter and it operates the propulsion system to perform orbit control. The nonlinear equation of the relative motion that regards the Earth as a point mass describes the relative states between 1U as the chief and 2U as the deputy. 

The configuration based on the mission concept is described in \figurename{\ref{fig:dstConfig}}. The first phase of the CANYVAL project is CANYVAL-X (CANYVAL-eXperimental) launched in 2018, and the OSC controls only the alignment angle ($\theta_{align}$) in \figurename{\ref{fig:dstConfig}} so that it has a flexible separation between DSC and OSC \cite{CANYVAL-X}. The second phase of the CANYVAL project is CANYVAL-C (CANYVAL-Coronograph) launched in 2021, and the thruster has a greater number of degrees of freedom to implement the distributed space telescope with a fixed distance of 40 meters \cite{CANYVAL-C}. Therefore, the formation flying system controls the error between the relative position of OSC (\begin{math}\textbf{\textit{r}}\end{math}) and the desired trajectory (\begin{math}\textbf{\textit{r}}_{r}\end{math}) to complete the DSC – OSC – Sun configuration.

\begin{figure}[b]
\begin{center}
\includegraphics[width=1.0\linewidth]{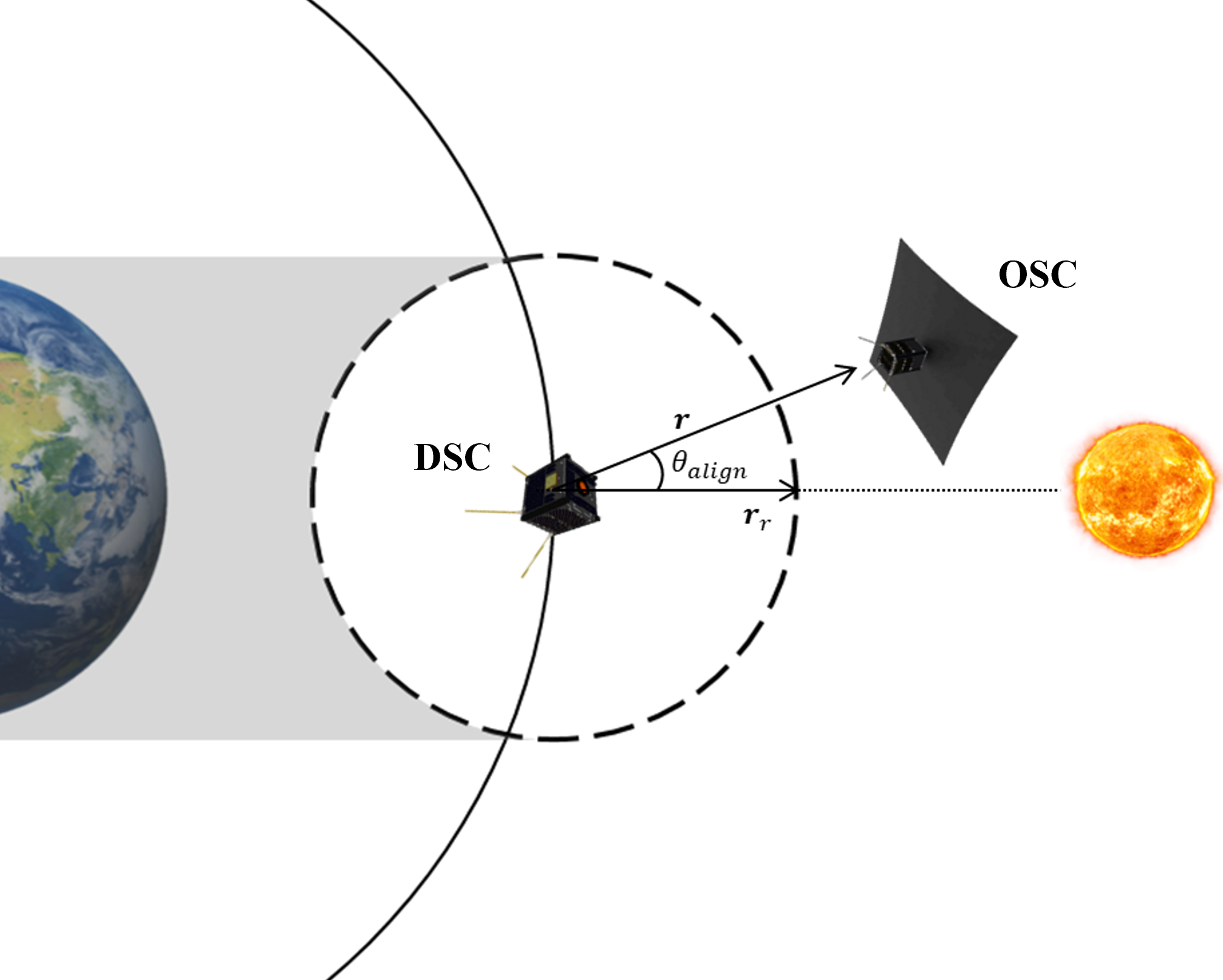}
\end{center}
\caption{Distributed space telescope configuration to demonstrate the solar coronagraph}
\label{fig:dstConfig}
\end{figure}

\figurename{\ref{fig:hwConfig}} describes the location and orientation of the payload and the propulsion system of the OSC. The CANYVAL-C mission's propulsion system utilizes GomSpace's Nanoprop, which fires cold gas through four nozzles in the \begin{math}\pm\hat{\textbf{\textit{b}}}_{x}\end{math} and \begin{math}\pm\hat{\textbf{\textit{b}}}_{y}\end{math} directions in the body fixed frame. The occulter should point to the \begin{math}-\hat{\textbf{\textit{b}}}_{z}\end{math} direction that is normal to the thrust nozzles and it constrains the controllable directions in the local-vertical, local-horizontal (LVLH) frame during alignment to the sun. The propulsion system controls only the satellite's orbit, while the attitude is controlled by a reaction wheel assembly composed of three wheels. The attitude should be controlled to fire in the desired direction as the propulsion system has two degrees of freedom. Therefore, the orbit and attitude control accuracies are dependent on each other. 

\begin{figure}[b]
\begin{center}
\includegraphics[width=1.0\linewidth]{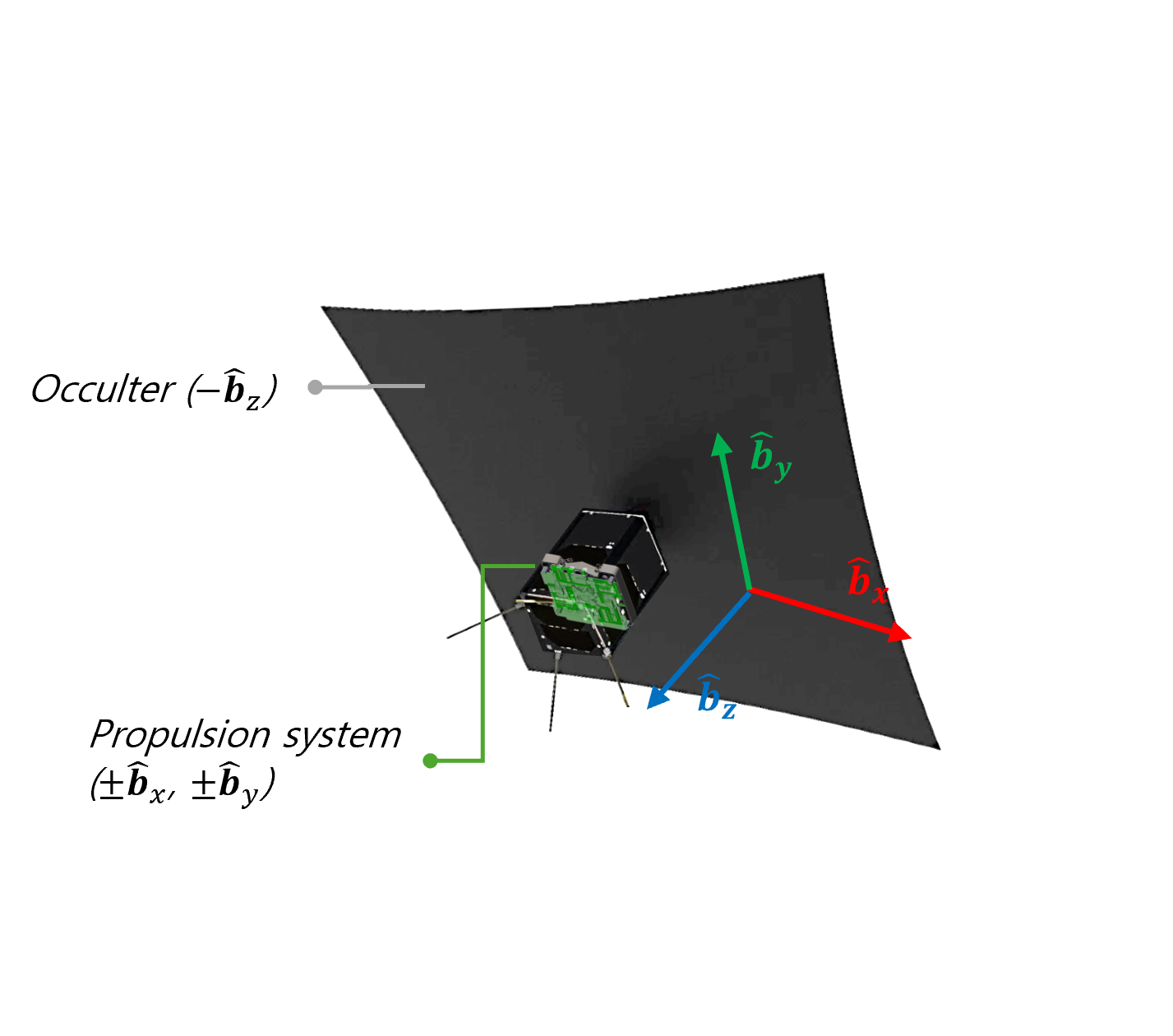}
\end{center}
\caption{Location and orientation of the occulter and propulsion system in body frame.}
\label{fig:hwConfig}
\end{figure}

\section{Description of Formation Flying System}
\label{s:3FFS}

This paper assumes that the propulsion system has two degrees of freedom and controls only the orbital states, and the ACS (Attitude Control System) consists of reaction wheels. Therefore, the attitude actuator controls the body frame so that the resultant fired thrust vector is aligned with the reference thrust vector in advance of the orbital actuator's firing. The relative orbital dynamics are described in the LVLH coordinate system, defined based on point mass gravity. The quaternions and angular velocity represent the attitude and the dynamic equations include the rotational motion of the reaction wheels.

\textcolor{black}{Section~\ref{s:3FFS}.\ref{s:3-1ROCS} and \ref{s:3-2ACS} introduce the main control plants of this paper, the relative orbit control system and attitude control system. The kinematic and dynamic equations of ROCS are \eqref{eq:r_rel}, \eqref{eq:twoBody_r_rel}, \eqref{eq:dynOrb}, and \eqref{eq:dynOrb_A1A2Eb}. The relative orbit error is defined as \eqref{eq:errorb} and follows the dynamics of \eqref{eq:dynOrbErr}. The kinematic and dynamic equations of ACS are expressed as \eqref{eq:kmtAtt}, \eqref{eq:qcrss} and \eqref{eq:dynAtt}. The quaternion and angular velocity errors are obtained as \eqref{eq:qerrAtt} and \eqref{eq:omerrAtt}, and the error dynamic equation is described as \eqref{eq:dynAttErr}. Section~\ref{s:3FFS}.\ref{s:3-3Problem} defines the problem that this paper addresses. Section ~\ref{s:3FFS}.\ref{s:3-4MissionObjectives} explains the reference trajectories and attitudes for the problem.}

\subsection{Relative Orbit Control System}
\label{s:3-1ROCS}
 The Keplerian two-body equations of motion of the chief satellite are given by \cite{KTAlfriend}
\begin{equation}{\label{eq:twoBody}}
    \Ddot{\textbf{\textit{r}}}_c \!\left(t\right) = - \frac{\mu\textbf{\textit{r}}_c\!\left(t\right)}{r_{c} ^3}
\end{equation}
where \begin{math} \textbf{\textit{r}}_c\!\left(t\right) = 
    \begin{bmatrix}
        X_{c}\!\left(t\right) & Y_{c}\!\left(t\right) & Z_{c}\!\left(t\right)
    \end{bmatrix}^T
\end{math} is the chief's position in an inertial frame with the magnitude of $r_{c}$ and \begin{math}\mu\end{math} is the standard gravitational parameter of the Earth.

The relative position and acceleration of the deputy satellite with respect to the chief are defined as
\begin{equation}{\label{eq:r_rel}}
    \textbf{\textit{r}}\!\left(t\right) = \textbf{\textit{r}}_d\!\left(t\right) - \textbf{\textit{r}}_c\!\left(t\right)
\end{equation}
\begin{equation}{\label{eq:twoBody_r_rel}}
\begin{aligned}
    &\Ddot{\textbf{\textit{r}}}\!\left(t\right) = \Ddot{\textbf{\textit{r}}}_d\!\left(t\right) - \Ddot{\textbf{\textit{r}}}_c\!\left(t\right) \\
    &=-\frac{\mu\left(\textbf{\textit{r}}\!\left(t\right) + \textbf{\textit{r}}_c\!\left(t\right)\right)}
    {\left\lVert \textbf{\textit{r}}\!\left(t\right) + \textbf{\textit{r}}_{c}\!\left(t\right) \right\lVert_{2}^3} 
    + \frac{\mu\textbf{\textit{r}}_c\!\left(t\right)}{r_c^3}
\end{aligned}
\end{equation}
where \begin{math} \textbf{\textit{r}}_d\!\left(t\right)=
    \begin{bmatrix}
        X_{d}\!\left(t\right) & Y_{d}\!\left(t\right) & Z_{d}\!\left(t\right)
    \end{bmatrix}^T
\end{math} is the inertial position of the deputy satellite with the distance from the center of the Earth $r_{d}$.

The second time derivative of the position vector in the inertial frame (\begin{math}i\end{math})
is related to the one in the LVLH frame (\begin{math}l\end{math}) through the following \cite{KTAlfriend,WEWiesel}
\begin{equation}{\label{eq:d2rdt2_NInertial}}
\begin{aligned}
    &\Ddot{\textbf{\textit{r}}}\!\left(t\right) = {^i}\left({\frac{d^{2}\textbf{\textit{r}}\!\left(t\right)}{dt^{2}}}\right) = 
    {^l}\left({\frac{d^{2}\textbf{\textit{r}}\!\left(t\right)}{dt^{2}}}\right)
    + 2{^i}\boldsymbol{\omega}{^l}\times {^l}\left(\frac{d\textbf{\textit{r}}\!\left(t\right)}{dt} \right) \\
    &+\frac{d{^i}\boldsymbol{\omega}{^l} }{dt}\times\textbf{\textit{r}}\!\left(t\right)
    +{^i}\boldsymbol{\omega}{^l}
    \times\left({^i}\boldsymbol{\omega}{^l} \times \textbf{\textit{r}}\!\left(t\right) \right)
\end{aligned}
\end{equation}
where \begin{math}{^i}\boldsymbol{\omega}{^l} \end{math}
is the angular velocity vector of the frame \begin{math}l\end{math}
with respect to the frame \begin{math}i\end{math}. It is defined as \begin{math}{^i}\boldsymbol{\omega}{^l} =
\begin{bmatrix}
0 & 0 & \dot{\theta}_c
\end{bmatrix}^T
\end{math} due to its normality to the orbital plane where \begin{math}
    \dot{\theta}_{c}
\end{math} is chief's the orbital angular speed.

Hence, the relative equations of motion are derived as
\begin{equation}{\label{eq:dynOrb}}
\begin{aligned}
    &\Ddot{\textbf{\textit{r}}}\!\left(t\right)=\textbf{A}_1\!\left(\textbf{\textit{r}}_{c},\dot{\textbf{\textit{r}}}_{c},\textbf{\textit{r}}_{d}\right) \textbf{\textit{r}} \!\left(t\right) + \textbf{A}_2\!\left(\textbf{\textit{r}}_{c},\dot{\textbf{\textit{r}}}_{c}\right) \dot{\textbf{\textit{r}}}\!\left(t\right) \\
    &+\textbf{\textit{f}}\left( \textbf{\textit{r}}_{c}, \textbf{\textit{r}}_{d} \right) + \frac{1}{m} \textbf{B} \left(\textbf{\textit{u}}_{orb}\!\left(t\right) + \textbf{\textit{d}}_{orb}\!\left(t,m,\textbf{\textit{r}}_{d},\dot{\textbf{\textit{r}}}_{d}\right)\right)
\end{aligned}
\end{equation}
where the nonlinearity of the Keplerian motion \begin{math}\textbf{\textit{f}}\left( \textbf{\textit{r}}_{c}, \textbf{\textit{r}}_{d} \right)\end{math}, the matrices \begin{math}\textbf{A}_1\!\left(\textbf{\textit{r}}_{c},\dot{\textbf{\textit{r}}}_{c},\textbf{\textit{r}}_{d}\right)\end{math}, \begin{math}\textbf{A}_2\!\left(\textbf{\textit{r}}_{c},\dot{\textbf{\textit{r}}}_{c}\right)\end{math}, and \begin{math}\textbf{B}\end{math} are defined as
\begin{equation}{\label{eq:dynOrb_A1A2Eb}}
    \begin{aligned}
    &\textbf{A}_1\!\left(\textbf{\textit{r}}_{c},\dot{\textbf{\textit{r}}}_{c},\textbf{\textit{r}}_{d}\right) = 
    \begin{bmatrix}
        \dot{\theta}_{c}^{2} - \frac{\mu}{r_d^3} & \Ddot{\theta}_{c} & 0 \\
        -\Ddot{\theta}_{c} & \dot{\theta}_{c}^2 - \frac{\mu}{r_d^3} & 0 \\
        0 & 0 & - \frac{\mu}{r_d^3}
    \end{bmatrix},\\
    &
    \textbf{A}_2\!\left(\textbf{\textit{r}}_{c},\dot{\textbf{\textit{r}}}_{c}\right) = 
    \begin{bmatrix}
        0 & 2\dot{\theta}_{c} & 0 \\
        -2\dot{\theta}_{c} & 0 & 0 \\
        0 & 0 & 0
    \end{bmatrix},\\
    &\textbf{\textit{f}}\left( \textbf{\textit{r}}_{c}, \textbf{\textit{r}}_{d}\right) =     
    \begin{bmatrix}
        \frac{\mu}{r_c^2}-\frac{\mu r_c}{r_d^3} \\ 0 \\ 0
    \end{bmatrix},\quad
    \textbf{B}=
    \begin{bmatrix}
        1 & 0 & 0\\
        0 & 1 & 0\\
        0 & 0 & 1
    \end{bmatrix}
    \end{aligned}
\end{equation}

\noindent and \begin{math}m\end{math} is the mass of the deputy, \begin{math}\theta_{c}\end{math}
is the argument of latitude of the chief, \begin{math}
    \textbf{\textit{u}}_{orb} \!\left(t\right) = 
    \begin{bmatrix}
        u_{orb,x}\!\left(t\right) & u_{orb,y}\!\left(t\right) & u_{orb,z}\!\left(t\right)
    \end{bmatrix}^T
\end{math} is the orbital control input, and \begin{math}
    \textbf{\textit{d}}_{orb}\!\left(t,m,\textbf{\textit{r}}_{d},\dot{\textbf{\textit{r}}}_{d}\right) =
    \begin{bmatrix}
        d_{orb,x}\!\left(t\right) & d_{orb,y}\!\left(t\right) & d_{orb,z}\!\left(t\right)
    \end{bmatrix}^T
\end{math} is the external disturbance force and its three elements have the same arguments but omitted here for brevity. 
The mass \begin{math}m\end{math} is assumed to have uncertainty as described in Section~\ref{s:4ctrlSys}.

Let the tracking error for the relative orbital states be defined by subtracting the reference state
\begin{math}\textbf{\textit{r}}_{r}\!\left(t\right)\end{math} from the current state \begin{math}\textbf{\textit{r}}\!\left(t\right)\end{math}:
\begin{equation}{\label{eq:errorb}}
\begin{aligned}
    &\textbf{\textit{r}}_e\!\left(t\right)=\textbf{\textit{r}}\!\left(t\right)-\textbf{\textit{r}}_{r}\!\left(t\right)\\
    &=
    \begin{bmatrix}
        x\!\left(t\right)-x_{r}\!\left(t\right) & y\!\left(t\right)-y_{r}\!\left(t\right) & z\!\left(t\right)-z_{r}\!\left(t\right)
    \end{bmatrix}^T
\end{aligned}
\end{equation}
where \begin{math}\textbf{\textit{r}}_e\!\left(t\right) =
    \begin{bmatrix}
        x_e\!\left(t\right) & y_e\!\left(t\right) & z_e\!\left(t\right)
    \end{bmatrix}^T
\end{math}.

Hence, the equations of motion of the orbital errors are derived by substituting \eqref{eq:errorb} into \eqref{eq:dynOrb} as
\begin{equation}{\label{eq:dynOrbErr}}
\begin{aligned}
    &\Ddot{\textbf{\textit{r}}}_e \!\left(t\right)=\textbf{A}_1\left(\textbf{\textit{r}}_{c},\dot{\textbf{\textit{r}}}_{c},\textbf{\textit{r}}_{d}\right)\left(\textbf{\textit{r}}_e\!\left(t\right)+\textbf{\textit{r}}_{r}\!\left(t\right)\right)\\
    &+\textbf{A}_2\left(\textbf{\textit{r}}_{c},\dot{\textbf{\textit{r}}}_{c}\right)\left(\dot{\textbf{\textit{r}}}_e\!\left(t\right)+\dot{\textbf{\textit{r}}}_{r}\!\left(t\right)\right)\\
    &+\textbf{\textit{f}} \left( \textbf{\textit{r}}_{c}, \textbf{\textit{r}}_{d}\right)
    +\frac{1}{m}\textbf{B}\left(\textbf{\textit{u}}_{orb}\!\left(t\right)+\textbf{\textit{d}}_{orb}\!\left(t,m,\textbf{\textit{r}}_{d},\dot{\textbf{\textit{r}}}_{d}\right)\right)
    -\Ddot{\textbf{\textit{r}}}_{r}\!\left(t\right)
\end{aligned}
\end{equation}
where \begin{math}\dot{\textbf{\textit{r}}}_e\!\left(t\right)\end{math} and 
\begin{math}\Ddot{\textbf{\textit{r}}}_e\!\left(t\right)\end{math}
are the first and second-time derivatives of the tracking error and 
\begin{math}\Ddot{\textbf{\textit{r}}}_{r}\!\left(t\right)\end{math}
is the second-time derivative of the reference trajectory.

\subsection{Attitude Control System}
\label{s:3-2ACS}

The kinematic differential equations for the quaternions and the angular velocity that represent the attitude control system of a satellite are written as follows \cite{JRWertz,BWie,PCHughes,PMTiwariAST,PMTiwari}
\begin{equation}{\label{eq:kmtAtt}}
\begin{aligned}
    \dot{\textbf{\textit{q}}}\!\left(t\right)
    &=
    \begin{bmatrix}
        \dot{\textbf{\textit{q}}}_{v}\!\left(t\right) \\ \dot{q}_{4}\!\left(t\right)
    \end{bmatrix} \\
    &=\frac{1}{2}
    \begin{bmatrix}
        q_{4}\!\left(t\right)\textbf{I}_{3\times3}+\textbf{\textit{q}}^{\times}\!\left(t\right) & \textbf{\textit{q}}_{v}\!\left(t\right) \\
        -\textbf{\textit{q}}_{v}^{T}\!\left(t\right) & q_{4}\!\left(t\right)
    \end{bmatrix}
    \begin{bmatrix}
        \boldsymbol{\omega}\!\left(t\right) \\ 0
    \end{bmatrix} 
\end{aligned}
\end{equation}
where the quaternion \begin{math}
    \textbf{\textit{q}}\!\left(t\right)=
    \begin{bmatrix}
        \textbf{\textit{q}}_{v}^{T}\!\left(t\right) & q_{4}\!\left(t\right)
    \end{bmatrix}^T
\end{math} consists of the vector \begin{math}\textbf{\textit{q}}_v\!\left(t\right)=
    \begin{bmatrix}
        q_{1}\!\left(t\right) & q_{2}\!\left(t\right) & q_{3}\!\left(t\right)
    \end{bmatrix}^{T}\end{math} and the scalar $q_{4}\!\left(t\right)$ components, \begin{math}\boldsymbol{\omega}\!\left(t\right)=
    \begin{bmatrix}
        \omega_{1}\!\left(t\right) & \omega_{2}\!\left(t\right) & \omega_{3}\!\left(t\right)
    \end{bmatrix}^{T}\end{math} is the angular velocity vector, \begin{math}\textbf{I}_{3\times3}\end{math} 
is the 3 by 3 identity matrix, and the skew-symmetric matrix \begin{math}\textbf{\textit{q}}^{\times}\!\left(t\right)\end{math}
is defined as 
\begin{equation}{\label{eq:qcrss}}
    \textbf{\textit{q}}^{\times}\!\left(t\right)=
    \begin{bmatrix}
        0 & -q_{3}\!\left(t\right) & q_{2}\!\left(t\right) \\
        q_{3}\!\left(t\right) & 0 & -q_{1}\!\left(t\right) \\
        -q_{2}\!\left(t\right) & q_{1}\!\left(t\right) & 0
    \end{bmatrix}
\end{equation}

The attitude dynamic equations of motion with the control torque generated by the reaction wheels under external disturbances can be expressed as
\begin{equation}{\label{eq:dynAtt}}
    \textbf{J}\dot{\boldsymbol{\omega}}\!\left(t\right)=-
    \boldsymbol{\omega}\!\left(t\right)\times\left(\textbf{J}\boldsymbol{\omega}\!\left(t\right)+\textbf{J}_{w}\boldsymbol{\omega}_{w}\!\left(t\right)\right)
    +\textbf{\textit{u}}_{att}\!\left(t\right)+\textbf{\textit{d}}_{att}\!\left(t\right)
\end{equation}
where $\textbf{J}$ is the deputy's moment of inertia matrix that contains uncertainty, $\textbf{J}_{w}$ is the moment of inertia matrix of the reaction wheels, $\boldsymbol{\omega}_{w}$ is the angular velocity of the reaction wheels, 
\begin{math}\textbf{\textit{u}}_{att}\!\left(t\right)\end{math} is the control torque, and
\begin{math}\textbf{\textit{d}}_{att}\!\left(t\right)\end{math} is the external disturbance torque.

The rotation formula from the current body axes to the reference body axes constructs the attitude tracking error quaternion \cite{JRWertz,BWie,PCHughes,PMTiwariAST,PMTiwari}
\begin{equation}{\label{eq:qerrAtt}}
\begin{aligned}
    \textbf{\textit{q}}_{e}\!\left(t\right)&=
    \begin{bmatrix}
        \textbf{\textit{q}}_{ev}\!\left(t\right) \\ q_{e4}\!\left(t\right)
    \end{bmatrix} \\
    &=
    \begin{bmatrix}
        q_{r4}\!\left(t\right)\textbf{\textit{q}}_{v}\!\left(t\right)-q_{4}\!\left(t\right)\textbf{\textit{q}}_{rv}\!\left(t\right)-\textbf{\textit{q}}_{rv}^{\times}\!\left(t\right)\textbf{\textit{q}}_{v}\!\left(t\right) \\
        q_{4}\!\left(t\right)q_{r4}\!\left(t\right)+\textbf{\textit{q}}_{rv}^{T}\!\left(t\right)\textbf{\textit{q}}_{v}\!\left(t\right)
    \end{bmatrix}
\end{aligned}
\end{equation}
where \begin{math}
    \textbf{\textit{q}}_{r}\!\left(t\right)=
    \begin{bmatrix}
        \textbf{\textit{q}}_{rv}^{T}\!\left(t\right) & q_{r4}\!\left(t\right)
    \end{bmatrix}^T
\end{math}
is the reference attitude trajectory and \begin{math}
    \textbf{\textit{q}}\!\left(t\right)=
    \begin{bmatrix}
        \textbf{\textit{q}}_{v}^{T}\!\left(t\right) & q_{4}\!\left(t\right)
    \end{bmatrix}^{T}
\end{math} represents the current attitude.

\textcolor{black}{
The rigid-body-attitude space, so-called $SO(3)$ is a boundaryless compact manifold and not a vector space. When a unit quaternion parametrizes $SO(3)$, its vector space $\mathbb{S}^{3}$ provides a double covering \cite{Unwind1},\cite{Unwind2}. Since $\textbf{\textit{q}}$ and $-\textbf{\textit{q}}$ represent the same orientation, there are two possible trajectories to connect two different quaternions. The termiology ‘unwinding phenomenon’ was first utilized in \cite{Unwind1} to describe the phenomenon that the spacecraft rotate along the longer path. This behavior produces stability of a single point in $\mathbb{S}^{3}$ and instability of the antipodal point at the same time, though two points physically represent the same orientation. Many previous studies address this phenomenon with hybrid-dynamic algorithm \cite{Unwind2} and modified sliding surface based on dual quaternions \cite{Unwind3}. This investigation prevents the unwinding phenomenon by choosing the sign of the scalar quaternion positive \cite{Unwind3}. }

The angular velocity error is derived as \cite{JRWertz,BWie,PCHughes,PMTiwariAST,PMTiwari}
\begin{equation}{\label{eq:omerrAtt}}
    \boldsymbol{\omega}_{e}\!\left(t\right)=\boldsymbol{\omega}\!\left(t\right)-\textbf{C}_{r}^{b}\!\left(t\right)\boldsymbol{\omega}_{r}\!\left(t\right)
\end{equation}
where \begin{math}
    \boldsymbol{\omega}\!\left(t\right)=
    \begin{bmatrix}
        \omega_{1}\!\left(t\right) & \omega_{2}\!\left(t\right) & \omega_{3}\!\left(t\right)
    \end{bmatrix}^T
\end{math} is the current angular velocity and \begin{math}
    \boldsymbol{\omega}_{r}\!\left(t\right)=
    \begin{bmatrix}
        \omega_{r1}\!\left(t\right) & \omega_{r2}\!\left(t\right) & \omega_{r3}\!\left(t\right)
    \end{bmatrix}^T
\end{math} is the reference angular velocity.
The matrix \begin{math}\textbf{C}_{r}^{b}\!\left(t\right)\end{math}
represents a rotation matrix from the reference to the current frames and can be expressed for the error quaternion as
\begin{equation}{\label{eq:q2rot}}
    \begin{aligned}
    \textbf{C}_{r}^{b}\!\left(t\right)&=\left(q_{e4}^{2}\!\left(t\right)-\textbf{\textit{q}}_{ev}^{T}\!\left(t\right)\textbf{\textit{q}}_{ev}\!\left(t\right)\right)\textbf{I}_{3\times3}
    +2\textbf{\textit{q}}_{ev}\!\left(t\right)\textbf{\textit{q}}_{ev}^{T}\!\left(t\right)\\
    &-2q_{e4}\!\left(t\right)\textbf{\textit{q}}_{ev}^{\times}\!\left(t\right)
    \end{aligned}
\end{equation}
where \begin{math}\left\lvert \textbf{C}_{r}^{b}\!\left(t\right) \right\rvert=1\end{math} and the derivative of the rotation matrix can derive the angular rate of the body fixed frame as
\begin{equation}{\label{eq:drotdt}}
    \dot{\textbf{C}}_{r}^{b}\!\left(t\right)=-\boldsymbol{\omega}_{e}^{\times}\!\left(t\right)\textbf{C}_{r}^{b}\!\left(t\right)
\end{equation}

Hence, the error kinematic and dynamic equations are written as \cite{JRWertz,BWie,PCHughes,PMTiwariAST,PMTiwari}
\begin{equation}{\label{eq:kmtAttErr}}
\begin{aligned}
    \dot{\textbf{\textit{q}}}_{e}\!\left(t\right)
    &=
    \begin{bmatrix}
        \dot{\textbf{\textit{q}}}_{ev}\!\left(t\right) \\ \dot{q}_{e4}\!\left(t\right)
    \end{bmatrix}\\
    &=\frac{1}{2}
    \begin{bmatrix}
    q_{e4}\!\left(t\right)\textbf{I}_{3\times3}+\textbf{\textit{q}}_{ev}^{\times}\!\left(t\right) & \textbf{\textit{q}}_{ev}\!\left(t\right)\\
        -\textbf{\textit{q}}_{ev}^{T}\!\left(t\right)    & q_{e4}\!\left(t\right)
    \end{bmatrix}
    \begin{bmatrix}
        \boldsymbol{\omega}_{e}\!\left(t\right) \\ 0
    \end{bmatrix} 
\end{aligned}
\end{equation}

\begin{equation}{\label{eq:dynAttErr}}
\begin{aligned}
    \textbf{J}\dot{\boldsymbol{\omega}}_{e} \!\left(t\right)
    &= -\left(\boldsymbol{\omega}_{e}\!\left(t\right) + \textbf{C}^{b}_{r}\!\left(t\right)\boldsymbol{\omega}_{r}\!\left(t\right)\right) \\
    &\times\left(\textbf{J}\left(\boldsymbol{\omega}_{e}\!\left(t\right)+\textbf{C}^{b}_{r}\!\left(t\right)\boldsymbol{\omega}_{r}\!\left(t\right)\right)+\textbf{J}_{w}\boldsymbol{\omega}_{w}\!\left(t\right) \right)\\
    &+ \textbf{\textit{u}}_{att}\!\left(t\right) + \textbf{\textit{d}}_{att}\!\left(t\right) \\
    &+\textbf{J}\left(\boldsymbol{\omega}_{e}^{\times}\!\left(t\right)\textbf{C}^{b}_{r}\!\left(t\right)\boldsymbol{\omega}_{r}\!\left(t\right)
    -\textbf{C}^{b}_{r}\!\left(t\right)\dot{\boldsymbol{\omega}}_{r}\!\left(t\right)\right)
\end{aligned}
\end{equation}

\subsection{Problem Statement}
\label{s:3-3Problem}

To summarize Sections \ref{s:2mission} and \ref{s:3FFS}, the relative orbit control problem of the CANYVAL-C mission is simplified into \textcolor{black}{a tracking problem}. The objective of this paper is to design a sliding surface and a control law for the formation flying system so that the errors of the relative orbit and attitude are ultimately bounded by small regions in finite time, that is:
\begin{equation}{\label{eq:ctrlObj_r}}
    \lim_{t \to\ t_{f,orb}} \left(\textbf{\textit{r}}\!\left(t\right)-\textbf{\textit{r}}_{r}\!\left(t\right)\right)=\textbf{\textit{R}}_{e}
\end{equation}
\begin{equation}{\label{eq:ctrlObj_qv}}
\begin{aligned}
    \lim_{t \to\ t_{f,att}} &\left(q_{r4}\!\left(t\right)\textbf{\textit{q}}_{v}\!\left(t\right)-q_{4}\!\left(t\right)\textbf{\textit{q}}_{rv}\!\left(t\right) \right.\\
    &\left.-\textbf{\textit{q}}_{rv}^{\times}\!\left(t\right)\textbf{\textit{q}}_{v}\!\left(t\right)\right)=\textbf{\textit{Q}}_{e}
\end{aligned}
\end{equation}
where the $t_{f,orb}$ and $t_{f,att}$ are the finite time instants of ROCS and ACS at which the errors converge to the desired regions. The error bounds $\textbf{\textit{R}}_{e}$ and $\textbf{\textit{Q}}_{e}$ are analytically derived in Lemma \ref{lemBoundOnRe} and \ref{lemBoundOnQe}, respectively. The errors are defined in \textcolor{black}{the next subsection.}

\textcolor{black}{
\subsection{Mission Objectives}
\label{s:3-4MissionObjectives}
}

The operational orbit is a sun-synchronous orbit (SSO), and the local time of the ascending node (LTAN) is 11:00. Since the sun-synchronous orbit is circular, the eccentricity of the chief's orbit is assumed to be zero. \textcolor{black}{The description about the inertial alignment hold in detail can be referred to \cite{SJeon}.} \textcolor{black}{The deputy commences the alignment toward the sun in the LVLH frame.} \textcolor{black}{Since the focal length of the coronagraph is 40 meters, the reference trajectory is derived as:}
\begin{equation}{\label{eq:refAlign}}
    \textbf{\textit{r}}_{r} (t)=
    40\cdot\frac{\textbf{\textit{r}}_{\odot} (t)}{\left\lVert{\textbf{\textit{r}}_{\odot} (t)}\right\rVert_{2}}\quad \left(\textnormal{m}\right)
\end{equation}
where \begin{math}\textbf{\textit{r}}_{\odot} (t)\end{math}
is the sun vector that orients from 1U to the sun and is described in the LVLH frame.

\figurename{\ref{fig:refTrj}} describes the reference orbital trajectories in the LVLH frame. \textcolor{black}{The epoch time and the LTAN determine the reference trajectory for the alignment because it depends on the position of the sun by \eqref{eq:refAlign}.} \figurename{\ref{fig:refTrj}} shows the reference trajectory of LTAN 11:00 at the vernal equinox when the sun is exactly at the equator in which the satellites were planned to operate for the CANYVAL-C mission.

\begin{figure}[t] 
\begin{center}
\includegraphics[width=1.0\linewidth]{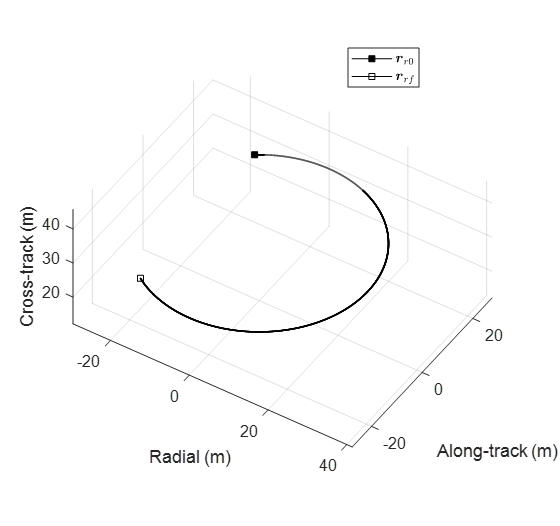}
\end{center}
\caption{Reference \textcolor{black}{trajectory} described in LVLH frame (\textcolor{black}{F}illed and empty squares: the initial and final reference states, solid line: the reference trajectory)}
\label{fig:refTrj}
\end{figure}

The reference attitude is derived from the TRIAD algorithm that defines a coordinate frame with two vectors described in the reference and body frames \cite{TRIAD1,JRWertz}:
\begin{equation}{\label{eq:rotTriad}}
    \textbf{C}^{b}\cdot\left(\textbf{C}^{i}\right)^{T}=\textbf{C}_{i}^{b}
\end{equation}
where the superscripts \begin{math}b\end{math} and \begin{math}i\end{math}
of the direction cosine matrix \begin{math}\textbf{C}\end{math}
refer to the body and reference frames, respectively. Also, the superscript \begin{math}T\end{math} denotes the transpose of the corresponding matrix. Given the primary and secondary vectors of TRIAD method 
\begin{math}\textbf{\textit{r}}_{p}\end{math} and 
\begin{math}\textbf{\textit{r}}_{s}\end{math}, the direction cosine matrix is derived as
\begin{equation}{\label{eq:rotVectors}}
    \textbf{C}=
    \begin{bmatrix}
        \hat{\textbf{\textit{r}}}_{1} & \hat{\textbf{\textit{r}}}_{2} & \hat{\textbf{\textit{r}}}_{3} 
    \end{bmatrix}
\end{equation}
where \begin{math}
    \hat{\textbf{\textit{r}}}_{1}={\textbf{\textit{r}}_{p}}/\left\lVert{\textbf{\textit{r}}_{p}}\right\rVert
\end{math}, 
\begin{math}
    \hat{\textbf{\textit{r}}}_{2}={\textbf{\textit{r}}}_{p}\times{\textbf{\textit{r}}}_{s}/\left\lVert {\textbf{\textit{r}}}_{p}\times{\textbf{\textit{r}}}_{s} \right\rVert
\end{math}, and
\begin{math}
    \hat{\textbf{\textit{r}}}_{3}= \textbf{\textit{r}}_{p}\times \left({\textbf{\textit{r}}}_{p}\times{\textbf{\textit{r}}}_{s}\right)
    /\left\lVert \textbf{\textit{r}}_{p}\times \left({\textbf{\textit{r}}}_{p}\times{\textbf{\textit{r}}}_{s}\right) \right\rVert
\end{math}.

The quaternion that is also known as the Euler parameters is defined as
\begin{equation}{\label{eq:qEulerParams}}
    \textbf{\textit{q}}_{v}=\textbf{\textit{a}}\sin\left(\theta/2\right),\quad
    q_{4}=\cos\left(\theta/2\right)
\end{equation}
where \begin{math}\textbf{\textit{a}}\end{math} and \begin{math}\theta
\end{math} are the Euler axis and the rotation angle about the Euler axis\cite{BWie,PCHughes}.

If \begin{math}q_{4}\end{math} is not zero, the quaternion is able to be derived from the direction cosine matrix as follows \cite{BWie,PCHughes}:
\begin{equation}{\label{eq:rot2qv}}
    \textbf{\textit{q}}_{v}=\frac{1}{4q_{4}}
    \begin{bmatrix}
        C_{23}-C_{32} & C_{31}-C_{13} & C_{12}-C_{21}
    \end{bmatrix}^{T}
\end{equation}
\begin{equation}{\label{eq:rot2q4}}
    \textit{q}_{4}= \textcolor{black}{\pm} \frac{1}{2}\left(1+C_{11}+C_{22}+C_{33}\right)^{\frac{1}{2}}
\end{equation}
where \begin{math}\textbf{C}=\left[C_{ij}\right]\end{math}.

The primary vector of the attitude control in the alignment mode is to orient the occulter toward the sun. Since the propulsion system has nozzles only along the $x$ and $y$ directions in the body frame, if it needs $z$ directional control, it rotates by 10 degrees with respect to the $y$ direction as an ad-hoc strategy. Therefore, the vectors that define the reference attitude frame during the alignment mode are:
\begin{equation}{\label{eq:rprsAlignInertial}}
    \textbf{\textit{r}}_{p}^{i}=\textbf{\textit{r}}_{\odot}^{i},\quad
    \textbf{\textit{r}}_{s}^{i}=\textbf{C}_{l}^{i}\cdot
    \begin{bmatrix}
    0 & 0 & 1
    \end{bmatrix}^{T}
\end{equation}
\begin{equation}{\label{eq:rpAlignBody}}
    \textbf{\textit{r}}_{p}^{b}=
    \begin{bmatrix}
        0 & \pm\sin{\theta} & -1-\cos{\theta}
    \end{bmatrix}^{T}
\end{equation}
\begin{equation}{\label{eq:rsAlignBody}}
    \textbf{\textit{r}}_{s}^{b}=
    \begin{bmatrix}
        1 & 0 & 0
    \end{bmatrix}^{T}
\end{equation}
where the subscripts $p$ and $s$ represent the primary and secondary vectors, respectively, \begin{math}\textbf{C}_{l}^{i}\end{math} is the rotation matrix from the LVLH frame to the inertial frame, 
\begin{math}\theta\end{math} is the rotation angle that is set as 10 degrees, and the positive sign in 
\begin{math}\textbf{\textit{r}}_{p}^{b}\end{math}
indicates that the $y$ and $z$ components of the thrust vector have the same sign. 
The secondary vectors in the inertial and body frames are derived from the 
angle \begin{math}\beta\end{math} of the sun-synchronous orbit that is the angle between the orbital plane and the sun vector. 

The deputy should be located within a sphere of a 3-meter radius with respect to the reference state to image the solar corona and determine the mission's success \cite{SJeon,CANYVAL-C}. The occulter should point to the sun with less than a 3-degree error to shade the sunlight in the image \cite{HKang,CANYVAL-C}. As a result, the requirements of ROCS and ACS are summarized as follows:
\begin{equation}\label{eq:reqROCS}
    \left\lVert \textbf{\textit{r}} - \textbf{\textit{r}}_{r} \right\lVert_{2} \leq 3 \quad\textnormal{(m)}
\end{equation}
\begin{equation}\label{eq:reqACS}
    \theta_{e} \leq 3 \quad\textnormal{(deg)}
\end{equation}
where $\theta_{e}$ is the Euler angle of the error quaternion in \eqref{eq:qEulerParams}.

\section{Control System Design}
\label{s:4ctrlSys}

\textcolor{black}{In this section, the assumptions about the observability and disturbances are explained, and then adaptive control law is introduced with Lemma~\ref{lemCtrlLaw} and \ref{lemAdtCtrlLaw}. In Section~\ref{s:4ctrlSys}.\ref{s:4-1ROCS}, \eqref{eq:sOrb} defines the sliding variable for ROCS, and Proposition \ref{proROCSErr} proves its stability. Theorem \ref{thmROCSCtrl} proves that the adaptive sliding mode control of \eqref{eq:ctrlLaw}, \eqref{eq:adtLaw}, and \eqref{eq:sOrb} is stable for ROCS. Lemma~\ref{lemBoundOnRe} derives the analytic orbital error bound. Section~\ref{s:4ctrlSys}.\ref{s:4-2ACS} follows the same structure as the previous section as the sliding variable for ACS is defined in \eqref{eq:sAtt}, Proposition \ref{proACSErr} proves its stability, and then Theorem \ref{thmACSCtrl} shows that the adaptive sliding mode control of \eqref{eq:ctrlLaw}, \eqref{eq:adtLaw}, and \eqref{eq:sAtt} guarantees the stability of ACS. Lemma~\ref{lemBoundOnQe} derives the attitude error bound in the same manner as Lemma~\ref{lemBoundOnRe}.}

\emph{\textbf{Assumption 1.}}
The state vectors of the satellite (\begin{math}
    \textbf{\textit{r}}, \dot{\textbf{\textit{r}}}, \textbf{\textit{q}}_{v}, \boldsymbol{\omega}
\end{math}) are fully measurable and available throughout the spaceflight mission.

\emph{\textbf{Assumption 2.}}
The uncertain structural properties in the orbit and attitude control systems are the mass and moments of inertia that have unknown lower and upper bounds. In the current paper, the mass and moments of inertia are supposed to be uncertain but bounded as
\begin{equation}{\label{eq:uncertainMass}}
    m\!\left(t\right) = m_{0} + \delta m\!\left(t\right)
\end{equation}
\begin{equation}{\label{eq:uncertainMoI}}
    \textbf{J}\!\left(t\right) = \textbf{J}_{0} + \delta\textbf{J}\!\left(t\right)
\end{equation}
where the subscripts `0' and `\begin{math}\delta\end{math}' denote the nominal (known) and uncertain values of the variable, respectively.


\emph{\textbf{Assumption 3.}}
The external disturbances are perturbations that include the asymmetric gravity of the Earth, atmospheric drag, and solar radiation pressure. They induce disturbance forces and torques that are bounded by unknown values:
\begin{equation}{\label{eq:distOrb}}
    0 < d_{f} < \left\lVert \textbf{\textit{d}}_{orb}\!\left(t,m,\textbf{\textit{r}}_{d},\dot{\textbf{\textit{r}}}_{d}\right) \right\rVert_{2} < D_{f}
\end{equation}
\begin{equation}{\label{eq:distAtt}}
    0 < d_{\tau} < \left\lVert \textbf{\textit{d}}_{att}\!\left(t\right) \right\rVert_{2} < D_{\tau}
\end{equation}
where (\begin{math}d_{f}\end{math}, \begin{math}D_{f}\end{math}) and (\begin{math}d_{\tau}\end{math}, \begin{math}D_{\tau}\end{math}) are the lower and upper bounds of the orbital and attitude disturbances, respectively, and are all unknown. 

\emph{\textbf{Adaptive control law}}

The adaptive smooth control law is based on sliding modes and drives the sliding variable into a user-specified domain without any chattering phenomenon. The robust control law designed for a SISO system in \cite{HCho2020} is elaborately extended to each of the orbit and attitude MIMO control systems:
\begin{equation}{\label{eq:ctrlLaw}}
    \textbf{\textit{u}}(t)=-\frac{K(t)}{\varepsilon}\textbf{\textit{s}}(t)
\end{equation}
where \begin{math}K(t)\end{math} is the smooth time-varying gain that obeys the adaptive law defined in \eqref{eq:adtLaw}, \begin{math}    \textbf{\textit{s}}(t)\end{math} is the sliding variable, and \begin{math}\varepsilon\end{math} is a user-specified positive constant. 

The gain is updated by the adaptation law as follows \cite{HCho2020}:
\begin{equation}\label{eq:adtLaw}
    \dot{K}(t)=\eta\left(\left\lVert \textbf{\textit{u}}(t) \right\rVert_{2}-K(t) + K_{0} \right),  K(0)\geq K_{0}
\end{equation}
where \begin{math}\textbf{\textit{u}}(t)\end{math} is the adaptive smooth control law in \eqref{eq:ctrlLaw}, 
\begin{math}\eta\end{math} is a user-specified positive constant, and \begin{math}K_{0}\end{math}
is the constant lower bound of the control gain \begin{math}K(t)\end{math} (see Lemma~\ref{lemAdtCtrlLaw}) that is positive and user-specified, and 
\begin{math}\left\lVert \cdot \right\rVert_{2}\end{math} is a vector 2-norm.

\begin{lemma}\label{lemCtrlLaw}
    Assume that the time derivative of the sliding variable can be expressed as
    \begin{equation}\label{eq:tmp27}
        \dot{\textbf{\textit{s}}}(t) = \Tilde{\lambda}_{max}^{-1}\Tilde{\textbf{\textit{d}}}(t) + \Tilde{\textbf{\textup{b}}}_{0}^{-1}\textbf{\textit{u}}(t)
    \end{equation}
    where $\Tilde{\textbf{\textup{b}}}_{0}$ is the nominal control matrix which is real-symmetric and positive definite, $\Tilde{\lambda}_{max}$ is the maximum eigenvalue of $\Tilde{\textbf{\textup{b}}}_{0}$, and $\Tilde{\textbf{\textit{d}}}(t)$ is the disturbance that has unknown bounds such that $\Tilde{d}<\left\lVert \Tilde{\textbf{\textit{d}}}(t) \right\lVert_{2}<\Tilde{D}$.
    For example, $\Tilde{\textbf{\textup{b}}}_{0}^{-1}$ is derived to be $\textbf{\textup{J}}_{0}^{-1}$ for ACS in this Section. $\Tilde{\textbf{\textit{d}}}(t)$ may include other terms than the dynamic forces and torques of Assumption 3 if it is bounded. Assume that the control gain $K(t)$ satisfies $K(t)>\Tilde{D}$ for $\forall t$. Then, as soon as the condition $\left\lVert\textbf{\textit{s}}\right\lVert_{2} > \varepsilon$ is met, the adaptive smooth control law proposed in \eqref{eq:ctrlLaw} makes the sliding variable converge into the region $\left\lVert \textbf{\textit{s}}\right\lVert_{2} \leq \varepsilon$ in finite time and retains it thereafter. 
    
\end{lemma}
\begin{proof}
    The Lyapunov candidate $V(t) = \frac{1}{2} \textbf{\textit{s}}^{T} \Tilde{\textbf{\textup{b}}}_{0}\textbf{\textit{s}}$ obtains its time derivative as
    \begin{equation}\label{eq:tmp28}
    \begin{aligned}
        \dot{V} &= \textbf{\textit{s}}^{T}\Tilde{\textbf{\textup{b}}}_{0}\dot{\textbf{\textit{s}}} \\
        &= \textbf{\textit{s}}^{T} \left( \Tilde{\lambda}_{max}^{-1}\Tilde{\textbf{\textup{b}}}_{0}\Tilde{\textbf{\textit{d}}} + \textbf{\textit{u}} \right) \\
        &= \textbf{\textit{s}}^{T} \left( \Tilde{\lambda}_{max}^{-1}\Tilde{\textbf{\textup{b}}}_{0}\Tilde{\textbf{\textit{d}}} - \frac{K(t)}{\varepsilon} \cdot\textbf{\textit{s}} \right)
    \end{aligned}
    \end{equation}
    where \eqref{eq:tmp27} and \eqref{eq:ctrlLaw} are substituted in order. 
    
    The spectral theorem decomposes any real symmetric matrix $\textbf{\textup{A}} \in \mathbb{R}^{n \times n}$ into
    \begin{equation}
    \begin{aligned}
        \textbf{A}&=\textbf{P}\textbf{D}\textbf{P}^{T} \\
        &= \sum_{i=1}^{n} \lambda_{i} \textbf{\textit{p}}_{i}\textbf{\textit{p}}_{i}^{T}
    \end{aligned}
    \end{equation}
    where $\lambda_{i}$ and $\textbf{\textit{p}}_{i}$ are the eigenvalue and eigenvector of $\textbf{\textup{A}}$ for $i=1,...,n$, $\textbf{\textup{P}}=
    \begin{bmatrix}
        \textbf{\textit{p}}_{1} & \textbf{\textit{p}}_{2} & ... & \textbf{\textit{p}}_{n}
    \end{bmatrix}$ is a matrix composed of eigenvectors, and $\textbf{\textup{D}} = diag\{\lambda_{1}, \lambda_{2}, ..., \lambda_{n} \}$ has the eigenvalues of $\textbf{\textup{A}}$ as its diagonal elements.
    In the same manner, $\textbf{\textit{s}}^{T}\Tilde{\textbf{\textup{b}}}_{0}\Tilde{\textbf{\textit{d}}}$ satisfies the inequality condition as
    \begin{equation}
        \textbf{\textit{s}}^{T}\Tilde{\textbf{\textup{b}}}_{0}\Tilde{\textbf{\textit{d}}} \leq \Tilde{\lambda}_{max} \textbf{\textit{s}}^{T} \Tilde{\textbf{\textit{d}}}   
    \end{equation}
    Therefore, \eqref{eq:tmp28} is derived as
    \begin{equation}
    \begin{aligned}
        \dot{V} &\leq \textbf{\textit{s}}^{T} \left( \Tilde{\textbf{\textit{d}}} - \frac{K(t)}{\varepsilon} \cdot \textbf{\textit{s}} \right) \\
        &< \left\lVert \textbf{\textit{s}} \right\lVert_{2} \Tilde{D} - \frac{K(t)}{\varepsilon} \left\lVert \textbf{\textit{s}} \right\lVert_{2}^{2} \\
        &= \left\lVert \textbf{\textit{s}} \right\lVert_{2} \left( \Tilde{D} - \frac{K(t)}{\varepsilon} \left\lVert \textbf{\textit{s}} \right\lVert_{2} \right)
    \end{aligned}
    \end{equation}
    If the sliding variable is in the region $\left\lVert \textbf{\textit{s}} \right\lVert_{2} > \varepsilon$, then
    \begin{equation}
    \begin{aligned}
        \dot{V} &< \left\lVert \textbf{\textit{s}} \right\lVert_{2} \left( \Tilde{D} - K(t) \right) \\
        &\leq - \kappa(t) \cdot V^{1/2}
    \end{aligned}
    \end{equation}
    where $\kappa(t)=\sqrt{2\lambda_{max}^{-1}} \left( K(t)-\Tilde{D} \right) > 0$.
    Therefore, the control law proposed in \eqref{eq:ctrlLaw} guarantees that the sliding variable moves into the region $\left\lVert \textbf{\textit{s}} \right\lVert_{2} \leq \varepsilon$ in finite time. 
    
\end{proof}

\begin{lemma}\label{lemAdtCtrlLaw}
    The gain adaptation law of \eqref{eq:adtLaw} implies that the control gain always satisfies $K(t) \geq K_{0}$ for $\forall t$.
\end{lemma}
\begin{proof}
    Let us define $\Delta K(t) \equiv K(t) - K_{0}$, then
    \begin{equation}\label{eq:tmp22}
        \Delta \dot{K}(t) = \dot{K}(t)
    \end{equation}
    Substitution of the adaptive rule into \eqref{eq:tmp22} leads to 
    \begin{equation}\label{eq:tmp23}
    \begin{aligned}
        \Delta \dot{K}(t) &= \dot{K} (t) \\
        &= \eta \left( \left\lVert \textbf{\textit{u}}(t) \right\lVert_{2} - K(t) + K_{0} \right) \\
        &= \eta \left( \left\lVert \textbf{\textit{u}}(t) \right\lVert_{2} - \Delta K(t) \right)
    \end{aligned}
    \end{equation}
    Since $\left\lVert \textbf{\textit{u}}(t) \right\lVert_{2} \geq 0$ is always satisfied, 
    \eqref{eq:tmp23} becomes
    \begin{equation}\label{eq:tmp25}
    \begin{aligned}
        \Delta \dot{K}(t) 
        &\geq -\eta \cdot \Delta K(t)
    \end{aligned}
    \end{equation}
    Let us introduce the function $\psi(t)=\psi(0)\exp{\left(-\eta t\right)}> 0$ where $\psi(0)$ is an arbitrary positive constant, then
    \begin{equation}
        \dot{\psi}(t) = -\eta\cdot\psi(t)
    \end{equation}
    The time derivative of $\frac{\Delta K(t)}{\psi(t)}$ is calculated as
    \begin{equation}\label{eq:tmp26}
    \begin{aligned}
        \frac{d}{dt} \left( \frac{\Delta K(t)}{\psi(t)} \right) &= \frac{\Delta \dot{K}(t) \psi(t) - \Delta K(t) \dot{\psi}(t)}{\psi^{2}(t)} \\
        &= \frac{\Delta \dot{K}(t) \psi(t) + \eta\cdot\Delta K(t)\psi(t)}{\psi^{2}(t)} \\
        &= \frac{\Delta \dot{K}(t) + \eta\cdot\Delta K(t)}{\psi(t)}
    \end{aligned}
    \end{equation}
    Since $\Delta \dot{K}(t) + \eta \cdot \Delta K(t) \geq 0$ by \eqref{eq:tmp25}, it proves that
    \begin{equation}\label{eq:KK_0ddtFunction}
        \frac{d}{dt}\left( \frac{\Delta K(t)}{ \psi(t)} \right) \geq 0.
    \end{equation}
    The initial conditions of $\Delta K(t)$ and $\psi(t)$
    yield 
    \begin{equation}\label{eq:KK_0initialFunction}
        \frac{\Delta K(0)}{\psi (0)} \geq 0
    \end{equation}
    where $\Delta K(0) = K(0) - K_{0} \geq 0$ and $\psi (0) > 0$.
    
    \noindent Consequently, \eqref{eq:KK_0ddtFunction} and \eqref{eq:KK_0initialFunction} proves that $\frac{\Delta K(t)}{\psi(t)}$ is equal to or greater than zero, leading to
    \begin{equation}
    \begin{aligned}
        &\Delta K(t) = K(t) - K_{0} \geq 0
    \end{aligned}
    \end{equation}
    Therefore, the gain is updated with the lower bound as $K(t) \geq K_{0}$.
\end{proof}

\begin{lemma}\label{Kbounded}
    The adaptive gain shown in \eqref{eq:adtLaw} has a finite upper bound such that $K \leq K^{*}$ and does not grow infinitely. It is proven in \cite{HCho2020}.
\end{lemma}

\subsection{Relative Orbit Control System}
\label{s:4-1ROCS}

\emph{\textbf{Sliding variable design}}

The sliding variable for the relative orbit control system is designed based on the concept of the nonsingular fast terminal sliding mode. This novel NFTSM variable guarantees the asymptotic stability to the equilibrium point while maintaining the fast convergent characteristics of the traditional TSM.
It is defined as follows:
\begin{equation}{\label{eq:sOrb}}
    \textbf{\textit{s}}_{orb}(t)
    =
    \dot{\textbf{\textit{r}}}_{e}(t)
    +
    \boldsymbol{\upalpha}_{orb}
    \textup{sig}^{\rho_{orb}}
    \textbf{\textit{r}}_{e}(t)
    +
    \boldsymbol{\upbeta}_{orb}
    \textbf{\textit{r}}_{e}(t)
\end{equation}
where \begin{math}\rho_{orb}\in(1,2)\end{math} is a constant, \begin{math}
    \boldsymbol{\upalpha}_{orb} = diag\left\{\alpha_{orb,1}, \alpha_{orb,2}, \alpha_{orb,3}\right\}
\end{math}
, \begin{math}
    \boldsymbol{\upbeta}_{orb} = diag\left\{\beta_{orb,1}, \beta_{orb,2}, \beta_{orb,3}\right\}
\end{math} 
are user-specified diagonal matrices, where 
\begin{math}\alpha_{orb,i=1,2,3}>0\end{math} and
\begin{math}\beta_{orb,i=1,2,3}>0\end{math}
are constants. Besides, \begin{math}\textnormal{sig}^{\rho}\boldsymbol{\xi}=
\begin{bmatrix}
{\left|\xi_{1} \right|}^{\rho}\textnormal{sgn}\left(\xi_{1}\right) & 
{\left|\xi_{2} \right|}^{\rho}\textnormal{sgn}\left(\xi_{2}\right) & 
{\left|\xi_{3} \right|}^{\rho}\textnormal{sgn}\left(\xi_{3}\right)
\end{bmatrix}^T
\end{math} is defined for a vector \begin{math}\boldsymbol{\xi}=
\begin{bmatrix}
\xi_{1} & \xi_{2} & \xi_{3}
\end{bmatrix}^T
\end{math}. The subscripts and arguments of the variables and parameters will be suppressed for brevity unless required.


\begin{proposition}\label{proROCSErr}
The relative position and velocity errors (\begin{math}\textbf{\textit{r}}_{e}, \dot{\textbf{\textit{r}}}_{e}\end{math}) are asymptotically stable if the sliding variable
\begin{math}\textbf{\textit{s}}\end{math} defined in \eqref{eq:sOrb} is zero.
\end{proposition}

\begin{proof}
Let the Lyapunov candidate be defined as
\begin{equation}{\label{eq:VOrb1}}
    V = \frac{1}{2}\textbf{\textit{r}}_{e}^{T}\textbf{\textit{r}}_{e}
\end{equation}

\noindent The time derivative is derived as
\begin{equation}{\label{eq:VOrbDot1}}
    \begin{aligned}
    \dot{V}
    &=\textbf{\textit{r}}_{e}^{T}\dot{\textbf{\textit{r}}}_{e}\\
    &=\textbf{\textit{r}}_{e}^{T}\left(-\boldsymbol{\upalpha}\textup{sig}^{\rho}\textbf{\textit{r}}_{e}-\boldsymbol{\upbeta}\textbf{\textit{r}}_{e}\right)\\    
    &\leq - \alpha_{min}\left\lVert\textbf{\textit{r}}_{e}\right\rVert_{\rho+1}^{\rho+1} - \beta_{min}\left\lVert\textbf{\textit{r}}_{e}\right\rVert_{2}^{2}
    \end{aligned}
\end{equation}
where $\textbf{\textit{s}} = \textbf{0}$ is used, \begin{math}\alpha_{min}=\min\limits_{i=1,2,3}\alpha_{i}\end{math} and \begin{math}\beta_{min}=\min\limits_{i=1,2,3}\beta_{i}\end{math}.

\noindent Hence, when the sliding variable is maintained to be zero, the orbital error states converge to zero asymptotically.
\end{proof}

\begin{lemma}\label{lemVdotV}

Assume that \begin{math}V\end{math} is a \begin{math}C^{1}\end{math} smooth positive definite function that is defined on \begin{math}U \subset \mathbb{R}^{n}\end{math}. If \begin{math}\dot{V} + cV^{r}\end{math} is a negative semi-definite function where \begin{math}c\in\mathbb{R}^{+}\end{math} and \begin{math}r \in (0,1)\end{math}, then there exists a region \begin{math}U_{0} \subset \mathbb{R}^{n}\end{math} such that \begin{math}V\end{math} can approach zero in finite time from the initial value \begin{math}V_{0}\end{math} on the region \begin{math}U_{0}\end{math}. The convergence time is known as \begin{math}t \leq \frac{V_{0}^{1-r}}{c(1-r)}\end{math}\cite{lemConv1,lemConv2,lemConv3,lemConv4,lemConv5}.

\end{lemma}

\begin{theorem}\label{thmROCSCtrl}
Consider the relative orbit control system that is nonlinear and uncertain \eqref{eq:dynOrb} with the adaptive smooth feedback controller \eqref{eq:ctrlLaw}-\eqref{eq:adtLaw} where the sliding manifold is designed as in \eqref{eq:sOrb}. Then, as soon as the condition \begin{math}\left\lVert\textbf{\textit{s}}\right\rVert_{2}>\varepsilon\end{math} is met, the orbital states converge to the region \begin{math}\left\lVert\textbf{\textit{s}}\right\rVert_{2}\leq\varepsilon\end{math} in finite time.
\end{theorem}

\begin{proof}
Let us define the Lyapunov function \begin{math}V(t)\end{math} as
\begin{equation}{\label{eq:VOrb2}}
    V=\frac{1}{2}\textbf{\textit{s}}^{T}\textbf{\textit{s}}
    +\frac{1}{2\gamma_{orb}}\left(K-K^{*}\right)^{2}
\end{equation}
where \begin{math}\gamma_{orb}\end{math} and \begin{math}K^{*}\end{math} are positive constants and \begin{math}K^{*}\end{math} is sufficiently large. The variable $\gamma$ refers to the orbital variable from here on.
The time derivative is derived as
\begin{equation}{\label{eq:VOrbDot2}}
    \dot{V}=\textbf{\textit{s}}^{T}\dot{\textbf{\textit{s}}}+\frac{1}{\gamma}\left(K-K^{*}\right)\dot{K}\\
\end{equation}
The time derivative of the orbital sliding variable yields
\begin{equation}{\label{eq:sOrbDot}}
    \dot{\textbf{\textit{s}}}=\ddot{\textbf{\textit{r}}}_{e}+ \boldsymbol{\upalpha}\rho\textbf{R}\dot{\textbf{\textit{r}}}_{e}+\boldsymbol{\upbeta}\dot{\textbf{\textit{r}}}_{e}
\end{equation}
where \begin{math}\textbf{R} = diag\left\{|x_{e}|^{\rho-1}, |y_{e}|^{\rho-1}, |z_{e}|^{\rho-1}\right\}
\end{math}. Substituting \eqref{eq:dynOrbErr} into \eqref{eq:sOrbDot} yields
\begin{equation}{\label{eq:sdottmp}}
\begin{aligned}
    \dot{\textbf{\textit{s}}}=&\textbf{A}_{1}\left(\textbf{\textit{r}}_{e}+\textbf{\textit{r}}_{r}\right)
    +\textbf{A}_{2}\left(\dot{\textbf{\textit{r}}}_{e}+\dot{\textbf{\textit{r}}}_{r}\right)
    +\textbf{\textit{f}}+\frac{1}{m}\textbf{B}\left(\textbf{\textit{u}} +\textbf{\textit{d}}\right)\\
    &-\ddot{\textbf{\textit{r}}}_{r}+\boldsymbol{\upalpha}\rho\textbf{R}\dot{\textbf{\textit{r}}}_{e}
    +\boldsymbol{\upbeta}\dot{\textbf{\textit{r}}}_{e}
\end{aligned}
\end{equation}
The terms in \eqref{eq:sdottmp} are rearranged according to their uncertainties
\begin{equation}
    \begin{aligned}
        \dot{\textbf{\textit{s}}}=&\left[\textbf{A}_{1}\left(\textbf{\textit{r}}_{e}+\textbf{\textit{r}}_{r}\right)
    +\textbf{A}_{2}\left(\dot{\textbf{\textit{r}}}_{e}+\dot{\textbf{\textit{r}}}_{r}\right)+\textbf{\textit{f}}
    -\ddot{\textbf{\textit{r}}}_{r}+\boldsymbol{\upalpha}\rho\textbf{R}\dot{\textbf{\textit{r}}}_{e}
    +\boldsymbol{\upbeta}\dot{\textbf{\textit{r}}}_{e}\right]\\
    &+\frac{1}{m}\textbf{B}\textbf{\textit{d}}+\frac{1}{m}\textbf{B} \textbf{\textit{u}}
    \end{aligned}
\end{equation}
With \eqref{eq:dynOrb_A1A2Eb}, the time derivative of the sliding variable simply leads to 
\begin{equation}
    \dot{\textbf{\textit{s}}}=\frac{1}{m}\textbf{\textit{u}}+\overline{\textbf{\textit{d}}}
\end{equation}
where \begin{math}
    \overline{\textbf{\textit{d}}}=[\textbf{A}_{1}\left(\textbf{\textit{r}}_{e}+\textbf{\textit{r}}_{r}\right)
    +\textbf{{A}}_{2}\left(\dot{\textbf{\textit{r}}}_{e}+\dot{\textbf{\textit{r}}}_{r}\right)
    +\textbf{\textit{f}}-\Ddot{\textbf{\textit{r}}}_{r}
    +\boldsymbol{\upalpha}\rho\textbf{R}\dot{\textbf{\textit{r}}}_{e}
    \end{math}
    \begin{math}
    +\boldsymbol{\upbeta}\dot{\textbf{\textit{r}}}_{e}]
    \end{math} includes the nonlinearities of the dynamics and uncertainties with an unknown upper bound
\begin{math}
    \left\lVert\overline{\textbf{\textit{d}}}\right\rVert_{2}<\overline{D}
\end{math}. 
Assume that \begin{math}K^{*}\end{math} is selected such that \begin{math}K^{*}>m_{max}\overline{D}\end{math} where $m_{max}$ is the maximum value of the uncertain mass and \begin{math}K^{*} \geq K(t)\end{math} for all time \begin{math}t\geq0\end{math} as \begin{math}K(t)\end{math} updated by \eqref{eq:adtLaw} is bounded (see Lemma \ref{Kbounded}), then
 \begin{equation}{\label{eq:tmp}}
 \begin{aligned}
     \dot{V}&=\textbf{\textit{s}}^{T}\dot{\textbf{\textit{s}}}
     +\frac{1}{\gamma}\left(K-K^{*}\right)\dot{K}\\
     &=\textbf{\textit{s}}^{T}\left(\overline{\textbf{\textit{d}}}+\frac{1}{m}\textbf{\textit{u}}\right)+\frac{1}{\gamma}\left(K-K^{*}\right)\dot{K}\\
     &=\textbf{\textit{s}}^{T}\left(\overline{\textbf{\textit{d}}}-\frac{K}{m\varepsilon}\textbf{\textit{s}}\right)+\frac{1}{\gamma}\left(K-K^{*}\right)\dot{K}\\
     &<\left\lVert\textbf{\textit{s}}\right\rVert_{2}\overline{D}-\frac{K}{m\varepsilon}\textbf{\textit{s}}^{T}\textbf{\textit{s}}+\frac{1}{\gamma}\left(K-K^{*}\right)\dot{K}\\
     &=\frac{\left\lVert\textbf{\textit{s}}\right\rVert_{2}}{m}\left(m\overline{D}-\frac{K}{\varepsilon}\left\lVert\textbf{\textit{s}}\right\rVert_{2}\right)+\frac{1}{\gamma}\left(K-K^{*}\right)\dot{K}
 \end{aligned}
 \end{equation}
 Let us assume the case where the sliding variable stays in the region
 \begin{math}\left\lVert\textbf{\textit{s}}\right\rVert_{2} > \varepsilon\end{math}.
 Then \eqref{eq:tmp} is derived as
 \begin{equation}
\begin{aligned}
     \dot{V}&<\frac{\left\lVert\textbf{\textit{s}}\right\rVert_{2}}{m}\left(m\overline{D}-K\right)+\frac{1}{\gamma}\left(K-K^{*}\right)\dot{K}\\
     &=\frac{\left\lVert\textbf{\textit{s}}\right\rVert_{2}}{m}\left(m\overline{D}-K^{*}\right)+\frac{1}{\gamma}\left(K-K^{*}\right)\dot{K}
     -\frac{\left\lVert\textbf{\textit{s}}\right\rVert_{2}}{m}\left(K-K^{*}\right)\\
     &=\left\lVert\textbf{\textit{s}}\right\rVert_{2}\left(\overline{D}-\frac{K^{*}}{m}\right)+\frac{1}{\gamma}\left(K-K^{*}\right)\dot{K}-\frac{\left\lVert\textbf{\textit{s}}\right\rVert_{2}}{m}\left(K-K^{*}\right)\\
     &<\left\lVert\textbf{\textit{s}}\right\rVert_{2}\left(\overline{D}-\frac{K^{*}}{m_{max}}\right)+\frac{1}{\gamma}\left(K-K^{*}\right)\dot{K}-\frac{\left\lVert\textbf{\textit{s}}\right\rVert_{2}}{m}\left(K-K^{*}\right)\\
     &=-\lambda_{s}\left\lVert\textbf{\textit{s}}\right\rVert_{2}-\left\lvert K-K^{*}\right\rvert\left(-\frac{\left\lVert\textbf{\textit{s}}\right\rVert_{2}}{m}+\frac{1}{\gamma}\dot{K}\right)
\end{aligned}
 \end{equation}
 where \begin{math}\lambda_{s}\triangleq\frac{K^{*}-m_{max}\overline{D}}{m_{max}}>0\end{math} is defined and the equality condition \begin{math}K-K^{*}=-\left\lvert K-K^{*} \right\rvert\end{math} is utilized.
 Now, a new parameter \begin{math}\lambda_{K}>0\end{math} is introduced such that
 \begin{equation}{\label{eq:tmp2}}
     \begin{aligned}
         \dot{V}&<-\lambda_{s}\left\lVert\textbf{\textit{s}}\right\rVert_{2}-\left\lvert K-K^{*}\right\rvert\left(-\frac{\left\lVert\textbf{\textit{s}}\right\rVert_{2}}{m}+\frac{1}{\gamma}\dot{K}\right)+\lambda_{K} \left\lvert K-K^{*}\right\rvert\\
         &\quad-\lambda_{K} \left\lvert K-K^{*}\right\rvert\\
         &=-\lambda_{s}\left\lVert\textbf{\textit{s}}\right\rVert_{2}-\left\lvert K-K^{*}\right\rvert
         \left(-\frac{\left\lVert\textbf{\textit{s}}\right\rVert_{2}}{m}+\frac{1}{\gamma}\dot{K}-\lambda_{K}\right)\\
         &\quad-\lambda_{K}\left\lvert K-K^{*}\right\rvert\\
         &<-\lambda_{s}\left\lVert\textbf{\textit{s}}\right\rVert_{2}-\left\lvert K-K^{*}\right\rvert\left(-\frac{\left\lVert\textbf{\textit{s}}\right\rVert_{2}}{m_{min}}+\frac{1}{\gamma}\dot{K}-\lambda_{K}\right)\\
         &\quad-\lambda_{K}\left\lvert K-K^{*}\right\rvert\\
         &=-\lambda_{s}\left\lVert\textbf{\textit{s}}\right\rVert_{2}-\lambda_{K}\left\lvert K-K^{*}\right\rvert-L
     \end{aligned}
 \end{equation}
 where \begin{math}
     L\triangleq\left\lvert K-K^{*}\right\rvert\left(-\frac{\left\lVert\textbf{\textit{s}}\right\rVert_{2}}{m_{min}}+\frac{1}{\gamma}\dot{K}-\lambda_{K}\right)
 \end{math}.
 We write \eqref{eq:tmp2} as
 \begin{equation}
     \begin{aligned}
         \dot{V}&<-\lambda_{s}\left\lVert\textbf{\textit{s}}\right\rVert_{2}-\lambda_{K} \left\lvert K-K^{*}\right\rvert-L\\
         &=-\lambda_{s}\sqrt{2}\cdot\frac{\left\lVert\textbf{\textit{s}}\right\rVert_{2}}{\sqrt{2}}
         -\lambda_{K}\sqrt{2\gamma}\cdot\frac{\left\lvert K-K^{*}\right\rvert}{\sqrt{2\gamma}}-L\\
         &\leq-\min\left(\lambda_{s}\sqrt{2},\lambda_{K}\sqrt{2\gamma}\right)\cdot
         \left(\frac{\left\lVert\textbf{\textit{s}}\right\rVert_{2}}{\sqrt{2}}+\frac{\left\lvert K-K^{*}\right\rvert}{\sqrt{2\gamma}}\right)-L\\
         &\leq-\lambda\cdot V^{1/2}-L
     \end{aligned}
 \end{equation}
 where \begin{math}
     \lambda\triangleq\sqrt{2}\min\left(\lambda_{s},\lambda_{K}\sqrt{\gamma}\right)>0
 \end{math}.
 
The parameter \begin{math}L\end{math} could be positive by selecting a proper \begin{math}\gamma\end{math}. Since $\dot{K} > 0$ when $\left\lVert \textbf{\textit{s}} \right\lVert_{2} > \varepsilon$, the positiveness condition \begin{math}L>0\end{math} is equivalent to
 \begin{equation}
     \gamma_{orb}<\frac{m_{min}\dot{K}}{\lambda_{K}m_{min}+\left\lVert\textbf{\textit{s}}\right\rVert_{2}}
 \end{equation}
Let us define $\Gamma_{orb}(t) \triangleq \frac{m_{min}\dot{K} }{\lambda_{K} m_{min} + \left\lVert \textbf{\textit{s}} \right\lVert_{2}}$, then substitution of \eqref{eq:adtLaw} leads to
\begin{equation}
\begin{aligned}
    \Gamma_{orb}(t) &\triangleq \eta \frac{ m_{min} \left( K(t)\left( \frac{\left\lVert \textbf{\textit{s}} \right\lVert_{2}}{\varepsilon} -1 \right) +K_{0} \right)}{\lambda_{K} m_{min} + \left\lVert \textbf{\textit{s}} \right\lVert_{2}} \\
    &\geq \eta \frac{m_{min} K_{0} \frac{\left\lVert \textbf{\textit{s}}\right\lVert_{2}}{\varepsilon}}{\lambda_{K}m_{min} + \left\lVert \textbf{\textit{s}} \right\lVert_{2}}\\
    &=\eta \cdot \frac{m_{min} \frac{K_{0}}{\varepsilon}}{\frac{\lambda_{K}m_{min}}{\left\lVert \textbf{\textit{s}} \right\lVert_{2}}+1}
\end{aligned}
\end{equation}
As the condition $\left\lVert \textbf{\textit{s}}\right\lVert_{2} > \varepsilon$ is satisfied, 
\begin{equation}
    \Gamma_{orb}(t) > \eta \frac{m_{min}\frac{K_{0}}{\varepsilon}}{\frac{\lambda_{K}m_{min}}{\varepsilon}+1}
    = \eta \frac{m_{min}K_{0}}{\lambda_{K} m_{min} + \varepsilon}
\end{equation}
As a result, one should select $\gamma_{orb}$ such that
\begin{equation}
    \gamma_{orb}<\frac{\eta m_{min} K_{0}}{\lambda_{K}m_{min} + \varepsilon}
\end{equation}

\noindent If the parameter \begin{math}L\end{math} is always positive, we can conclude that \begin{math}
     \dot{V}<-\lambda V^{1/2}-L<-\lambda V^{1/2}
 \end{math}. Hence, the sliding surface converges into the region \begin{math}
     \left\lVert\textbf{\textit{s}}\right\rVert_{2} \leq\varepsilon
 \end{math} in finite time by Lemma \ref{lemVdotV}.  
\end{proof}

\begin{lemma}\label{lemBoundOnRe}
When the $i$th element of the orbital sliding variable is bounded by $\left\lvert s_{i} \left(t\right) \right\lvert \leq \varepsilon_{i}$, the relative position error is bounded by $\left\lvert r_{ei} \left(t\right) \right\lvert \leq R_{ei}$, where $x_{e}$, $y_{e}$, and $z_{e}$ are written as $r_{ei}$ for $i=1,2,3$, respectively, to possess the expression of uniformity. The error bound $R_{ei}$ can be analytically obtained, being the unique solution of the algebraic equation $F(R_{ei})=0$, where $F(x)=\varepsilon_{i} - \alpha_{i}\textnormal{sig}^{\rho}x - \beta_{i}x$. Because $R_{ei}$ is a positive constant, $\textnormal{sig}^{\rho}x$ is equivalent to $x^{\rho}$ and it leads to $F(x)=\varepsilon_{i} - \alpha_{i} x^{\rho} - \beta_{i}x$. \textcolor{black}{This relation enables the derivation of control parameters that are physically meaningful and consistent with the control requirement.}
\end{lemma}

\begin{proof}
        
    Let us define the Lyapunov function $V=\frac{1}{2}r_{ei}^{2}$, then its time derivative is derived as
    \begin{equation} \label{eq:tmp15}
        \dot{V}=r_{ei}\dot{r}_{ei}
    \end{equation}
    The bound on the sliding surface $\left|s_{i}\right|\leq\varepsilon_{i}$ constrains $\dot{r}_{ei}$ as
    \begin{equation} \label{eq:tmp18}
        G(r_{ei}) \leq \dot{r}_{ei} \leq F(r_{ei}) \\
    \end{equation}
    where $F(x)$ and $G(x)$ are defined as in \eqref{eq:tmp19} and \eqref{eq:tmp20}
    \begin{equation} \label{eq:tmp19}
        F(x) = \varepsilon_{i} - \alpha_{i} \textnormal{sig}^{\rho} x - \beta_{i} x
    \end{equation}
    \begin{equation} \label{eq:tmp20}
        G(x) = -\varepsilon_{i} - \alpha_{i} \textnormal{sig}^{\rho} x - \beta_{i}x
    \end{equation}
    The description of $F(x)$ depends on the sign of $x$ and achieves its derivative $\frac{dF(x)}{dx}$ as
    \begin{equation}\label{eq:lemOrbfx}
        F(x) = \begin{cases}
            \varepsilon_{i} + \alpha_{i}(-x)^{\rho} + \beta_{i}(-x) >0 & \text{if $x<0$}, \\
            \varepsilon_{i} >0 & \text{if $x=0$}, \\
             \varepsilon_{i} - \alpha_{i}x^{\rho} - \beta_{i}x & \text{if $x>0$}. \\
        \end{cases}
    \end{equation}
    \begin{equation}\label{eq:lemOrbfxPr}
        \frac{dF(x)}{dx} = - \rho\alpha_{i}\left|x\right|^{\rho-1} - \beta_{i} < 0 \quad\quad\quad \text{for $\forall x \in \mathbb{R}$}.
    \end{equation}
    Since $F(0)>0$, $F(x)$ is positive for $\forall x \in \mathbb{R}^{-}$ and monotonically decreasing for $\forall x \in \mathbb{R}$, it is obvious that there exits an unique positive solution $R_{ei,p} \in \mathbb{R}^{+}$ such that $F(R_{ei,p})=0$. 
    The sign of $F(x)$ is calculated as
    \begin{equation}\label{eq:tmp16}
        \begin{cases}
        F(x) > 0 & \text{for $x < R_{ei,p}$}, \\
        F(x) = 0 & \text{for $x = R_{ei,p} > 0$}, \\
        F(x) < 0 & \text{for $x > R_{ei,p} > 0$}. \\
        \end{cases}
    \end{equation}

    On the other hand, $G(x) = F(x) - 2 \varepsilon_{i}$, which means $G(x)$ is translated $F(x)$ along the $y$ direction by $-2\varepsilon_{i}$. 
    \begin{equation}\label{eq:lemOrbfx2}
        G(x) = \begin{cases}
            - \varepsilon_{i} + \alpha_{i}(-x)^{\rho} + \beta_{i}(-x) & \text{if $x<0$}, \\
            - \varepsilon_{i} < 0 & \text{if $x=0$}, \\
             - \varepsilon_{i} - \alpha_{i}x^{\rho} - \beta_{i}x < 0 & \text{if $x>0$}. \\
        \end{cases}
    \end{equation}
    \begin{equation}\label{eq:lemOrbfxPr2}
        \frac{dG(x)}{dx} = \frac{dF(x)}{dx} < 0.
    \end{equation}
    Likewise, the sign of $G(x)$ is determined according to $R_{ei,n}\in \mathbb{R}^{-}$ which is the unique solution of $G(R_{ei,n})=0$.
    \begin{equation}\label{eq:tmp17}
        \begin{cases}
        G(x) > 0 & \text{for $x < R_{ei,n} < 0$}, \\
        G(x) = 0 & \text{for $x = R_{ei,n} < 0$}, \\
        G(x) < 0 & \text{for $x > R_{ei,n}$}. \\
        \end{cases}
    \end{equation}
    As a result, substitution of \eqref{eq:tmp16}, \eqref{eq:tmp17} to \eqref{eq:tmp18} derives 
    \begin{equation}
        \begin{cases}
            G(r_{ei}) \leq \dot{r}_{ei} \leq F(r_{ei}) < 0 & \text{if $0 < R_{ei,p} < r_{ei}$} \\
            0 < G(r_{ei}) \leq \dot{r}_{ei} \leq F(r_{ei}) & \text{if $r_{ei} < R_{ei,n} < 0$} \\
        \end{cases}
    \end{equation}

    The parallel translation between $F(x)$ and $G(x)$ derives the relationship between $R_{ei,p}$ and $R_{ei,n}$ as
    \begin{equation}
    \begin{aligned}
        G(R_{ei,n}) &= - \varepsilon_{i} - \alpha_{i} \textnormal{sig}^{\rho} R_{ei,n} - \beta_{i} R_{ei,n} \\
        & = - \varepsilon_{i} + \alpha_{i} (-R_{ei,n})^{\rho} + \beta_{i} (-R_{ei,n}) \\
        & = - F(-R_{ei,n}) \\
        & = 0
    \end{aligned}
    \end{equation}
    Because $F(R_{ei,p})=F(- R_{ei,n})=0$ and $F(x)=0$ has the unique solution, it implies that $R_{ei,p} = - R_{ei,n} = R_{ei}$.
    The time derivative of the Lyapunov function in \eqref{eq:tmp15} become negative when $ \left| r_{ei} \right| > R_{ei}$ as
    \begin{equation}\label{eq:tmp21}
        \dot{V} = r_{ei} \dot{r}_{ei} < 0
    \end{equation}
    which proves that the region $\mathbb{L}_{1} = \{r_{ei} \in \mathbb{R}\, |\, \left| r_{ei} \right| \leq R_{ei}\}$ is the domain of attraction. \figurename{\ref{fig:boundOnRe}} explains this region of attraction as the shaded region satisfies \eqref{eq:tmp21}.
\begin{figure}[b]
\begin{center}
\includegraphics[width=0.9\linewidth]{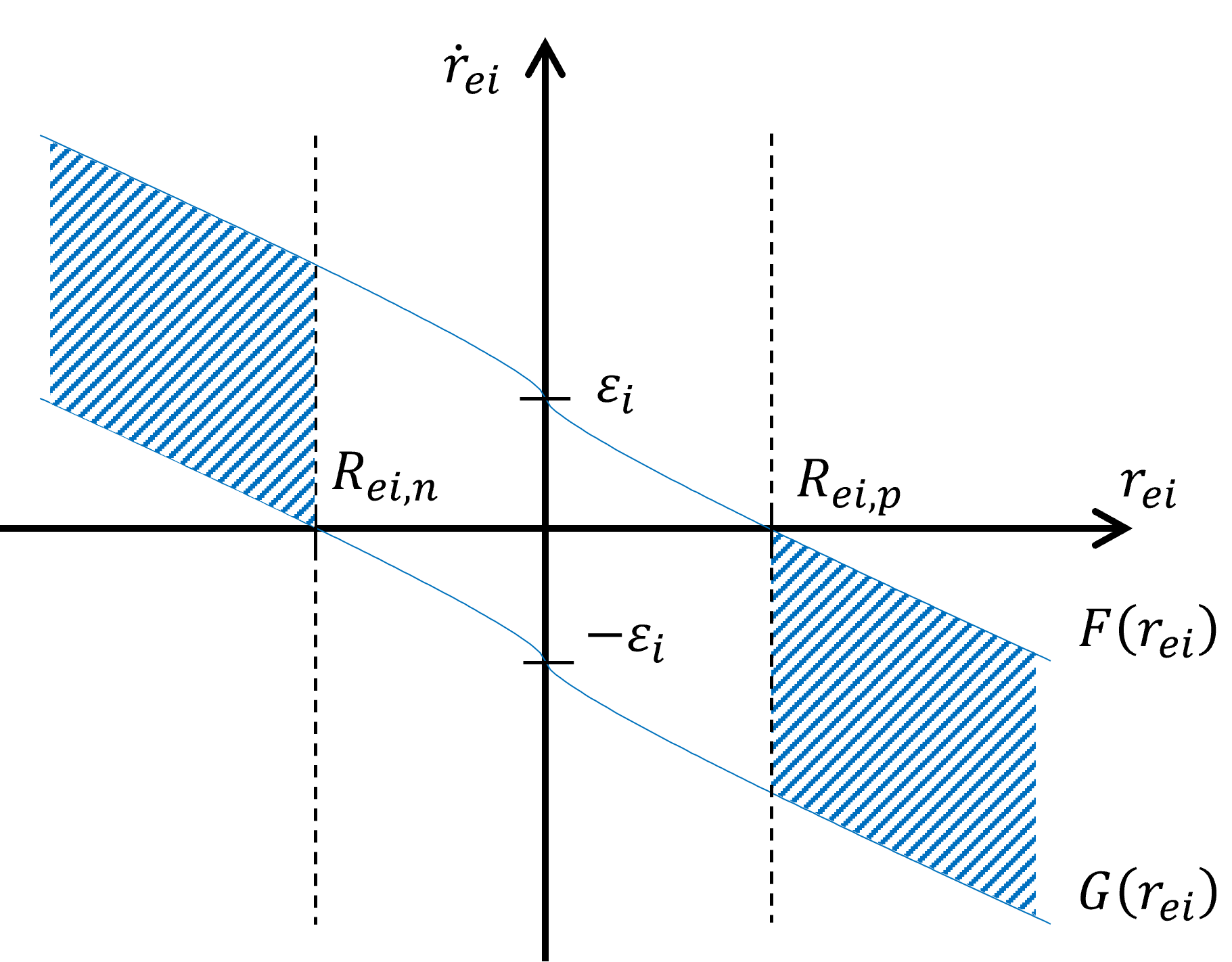}
\end{center}
\caption{Bound on $r_{ei}$ in the region $\left\lvert s_{i}\right\lvert \leq \varepsilon_{i}$ }
\label{fig:boundOnRe}
\end{figure}

\end{proof}

\subsection{Attitude Control System}
\label{s:4-2ACS}

\emph{\textbf{Sliding variable design}}

The sliding variable for the attitude control system is designed based on nonsingular fast terminal sliding mode as follows:
\begin{equation}{\label{eq:sAtt}}
    \textbf{\textit{s}}_{att}=\boldsymbol{\omega}_{e}
    +\boldsymbol{\upalpha}_{att}\textup{sig}^{\rho_{att}}\textbf{\textit{q}}_{ev}
    +\boldsymbol{\upbeta}_{att}\textbf{\textit{q}}_{ev}    
\end{equation}
where the definition of \begin{math}\textnormal{sig}^{\rho} \boldsymbol{\xi}\end{math}, $\boldsymbol{\upalpha}$, and $\boldsymbol{\upbeta}$ are the same with the one in \eqref{eq:sOrb}, but related to the ACS.

\begin{proposition}\label{proACSErr}
The quaternion and angular velocity errors (\begin{math}\textbf{\textit{q}}_{ev}, \boldsymbol{\omega}_{e}\end{math}) are asymptotically stable if the sliding variable
\begin{math}\textbf{\textit{s}}\end{math} in \eqref{eq:sAtt} is zero.
\end{proposition}

\begin{proof}

    Let us define the Lyapunov candidate as \cite{YJCheon}
    \begin{equation}
        V=\textbf{\textit{q}}_{ev}^{T}\textbf{\textit{q}}_{ev} + (1-q_{e4})^2,
    \end{equation} then its time derivative is derived by substituting \eqref{eq:kmtAttErr} as
    \begin{equation}{\label{eq:tmp3}}
        \dot{V}=-2\dot{q}_{e4}=\textbf{\textit{q}}_{ev}^{T}\boldsymbol{\omega}_{e}
    \end{equation}
    Since the attitude sliding variable is assumed to be zero, then it is derived by substituting \eqref{eq:sAtt}
    \begin{equation}
    \begin{aligned}
        \dot{V}&\leq \textbf{\textit{q}}_{ev}^{T}
        \left(-\boldsymbol{\upalpha}\textup{sig}^{\rho}\textbf{\textit{q}}_{ev}-\boldsymbol{\upbeta}\textbf{\textit{q}}_{ev}\right)\\
        &\leq-\alpha_{min}\textbf{\textit{q}}_{ev}^{T}\textup{sig}^{\rho}\textbf{\textit{q}}_{ev}
        -\beta_{min}\textbf{\textit{q}}_{ev}^{T}\textbf{\textit{q}}_{ev}\\
        &=-\alpha_{min}\left\lVert\textbf{\textit{q}}_{ev}\right\lVert_{\rho+1}^{\rho+1}
        -\beta_{min}\left\lVert\textbf{\textit{q}}_{ev}\right\lVert_{2}^{2}
    \end{aligned}
    \end{equation}
    where \begin{math}\alpha_{min}=\min\limits_{i=1,2,3}\alpha_{i}\end{math} and \begin{math}\beta_{min}=\min\limits_{i=1,2,3}\beta_{i}\end{math}.
    
    Therefore, when the attitude sliding variable is retained at zero, the quaternion error converges to zero. If the sliding variable and the quaternion error are zero, the angular velocity error is also zero according to the definition of the sliding variable in \eqref{eq:sAtt}. As a result, it proves the asymptotic stability of the quaternion and angular velocity errors when the sliding variable is zero. 
\end{proof}

\begin{theorem}\label{thmACSCtrl}
Consider the attitude control system that is nonlinear and uncertain \eqref{eq:dynAtt} with the adaptive smooth feedback control \eqref{eq:ctrlLaw} and \eqref{eq:adtLaw}. The sliding manifold is designed as \eqref{eq:sAtt}. Then, when the condition \begin{math}\left\lVert\textbf{\textit{s}}\right\rVert_{2}>\varepsilon\end{math} is met, the attitude states converge to the region \begin{math}\left\lVert\textbf{\textit{s}}\right\rVert_{2}\leq\varepsilon\end{math} in finite time.
\end{theorem}

\begin{proof}
    The Lyapunov candidate is defined as
\begin{equation}\label{eq:VAtt2}
    V=\frac{1}{2}\textbf{\textit{s}}^{T}\textbf{J}_{0}\textbf{\textit{s}}+\frac{1}{2\gamma_{att}}\left(K-K^{*}\right)^{2}
\end{equation}
where \begin{math}\gamma_{att}\end{math} is a positive constant and \begin{math}
    K^{*}
\end{math} is also a positive constant satisfying \begin{math}K(t) \leq K^{*}\end{math}. The variable $\gamma $ implies that it is the attitude variable from here on. 
By the definition of the sliding variable \eqref{eq:sAtt}, its time derivative is described as
\begin{equation}\label{eq:sAttDot}    \dot{\textbf{\textit{s}}}=\dot{\boldsymbol{\omega}}_{e}+\left(\rho\boldsymbol{\upalpha}\boldsymbol{\upTheta}+\boldsymbol{\upbeta}\right)\dot{\textbf{\textit{q}}}_{ev}
\end{equation}
where \begin{math}
    \boldsymbol{\upTheta}=diag\left\{\left\lvert q_{e1}\right\rvert^{\rho-1}, \left\lvert q_{e2}\right\rvert^{\rho-1}, \left\lvert q_{e3}\right\rvert^{\rho-1}\right\}
\end{math} is positive semi-definite.

The moments of inertia matrix is decomposed as in \eqref{eq:uncertainMoI}, and the inverse is derived as
\begin{equation}\label{eq:JInvDecompose}
\textbf{J}^{-1}=\left(\textbf{J}_{0}+\delta\textbf{J}\right)^{-1}=\textbf{J}_{0}^{-1}+\delta\hat{\textbf{J}}
\end{equation}
where \begin{math}
    \delta\hat{\textbf{J}}=-\textbf{J}_{0}^{-1}\delta\textbf{J}\left(\textbf{I}_{3\times3}+\textbf{J}_{0}^{-1}\delta\textbf{J}\right)^{-1}\textbf{J}_{0}^{-1}
\end{math} \cite{PMTiwariAST,KMa}.

\noindent Hence, the substitution of \eqref{eq:JInvDecompose} into \eqref{eq:dynAttErr} leads to the decomposition of $\dot{\boldsymbol{\omega}_{e}}$ in \eqref{eq:sAttDot} as 
\begin{equation}\label{eq:tmp4}  
    \begin{aligned}
    &\textbf{J}^{-1}\left[-\left(\boldsymbol{\omega}_{e} + \textbf{C}^{b}_{r}\boldsymbol{\omega}_{r}\right) \times
    \left(\textbf{J}\left(\boldsymbol{\omega}_{e}+\textbf{C}^{b}_{r}\boldsymbol{\omega}_{r}\right)
    +\textbf{J}_{w}\boldsymbol{\omega}_{w}\right) \right.\\
    &\left. \quad+\textbf{\textit{u}} + \textbf{\textit{d}}
    +\textbf{J}\left(\boldsymbol{\omega}_{e}^{\times}\textbf{C}^{b}_{r}\boldsymbol{\omega}_{r}
    -\textbf{C}^{b}_{r}\dot{\boldsymbol{\omega}}_{r}\right)\right]\\
    &=\left(\textbf{J}_{0}^{-1}+\delta\hat{\textbf{J}}\right)
    \left[-\left(\boldsymbol{\omega}_{e} + \textbf{C}^{b}_{r}\boldsymbol{\omega}_{r}\right) \times\left(\left(\textbf{J}_{0}+\delta\textbf{J}\right)\left(\boldsymbol{\omega}_{e}+\textbf{C}^{b}_{r}\boldsymbol{\omega}_{r}\right) \right.\right.\\
    &\left.\left. \quad+\textbf{J}_{w}\boldsymbol{\omega}_{w}\right)+\textbf{\textit{u}} + \textbf{\textit{d}}+\left(\textbf{J}_{0}+\delta\textbf{J}\right)\left(\boldsymbol{\omega}_{e}^{\times}\textbf{C}^{b}_{r}\boldsymbol{\omega}_{r}
    -\textbf{C}^{b}_{r}\dot{\boldsymbol{\omega}}_{r}\right) \right]\\
    &=\textbf{J}_{0}^{-1}\left[-\left(\boldsymbol{\omega}_{e} + \textbf{C}^{b}_{r}\boldsymbol{\omega}_{r}\right)\times\textbf{J}_{0}\left(\boldsymbol{\omega}_{e} + \textbf{C}^{b}_{r}\boldsymbol{\omega}_{r}\right) \right.\\
    &\left. \quad-\left(\boldsymbol{\omega}_{e} + \textbf{C}^{b}_{r}\boldsymbol{\omega}_{r}\right)\times\textbf{J}_{w}\boldsymbol{\omega}_{w}
    +\textbf{\textit{u}}+\textbf{J}_{0}\left(\boldsymbol{\omega}_{e}^{\times}\textbf{C}_{r}^{b}\boldsymbol{\omega}_{r}-\textbf{C}_{r}^{b}\boldsymbol{\omega}_{r}\right)\right]\\
    &+\textbf{J}_{0}^{-1}\left[-\left(\boldsymbol{\omega}_{e} + \textbf{C}^{b}_{r}\boldsymbol{\omega}_{r}\right)\times\delta\textbf{J}\left(\boldsymbol{\omega}_{e} + \textbf{C}^{b}_{r}\boldsymbol{\omega}_{r}\right) \right.\\ &\left.\quad+\textbf{\textit{d}}+\delta\textbf{J}\left(\boldsymbol{\omega}_{e}^{\times}\textbf{C}_{r}^{b}\boldsymbol{\omega}_{r}-\textbf{C}_{r}^{b}\boldsymbol{\omega}_{r}\right)\right]\\
    &+\delta\hat{\textbf{J}}\left[ - \left( \boldsymbol{\omega}_{e} + \textbf{C}^{b}_{r}\boldsymbol{\omega}_{r} \right) \times\textbf{J}_{0}
    \left( \boldsymbol{\omega}_{e} + \textbf{C}^{b}_{r}\boldsymbol{\omega}_{r} \right) \right.\\
    &\left.\quad- \left( \boldsymbol{\omega}_{e} + \textbf{C}^{b}_{r}\boldsymbol{\omega}_{r} \right)\times\delta\textbf{J} \left( \boldsymbol{\omega}_{e} + \textbf{C}^{b}_{r}\boldsymbol{\omega}_{r} \right) \right.\\ 
    &\left.\quad-\left(\boldsymbol{\omega}_{e} + \textbf{C}^{b}_{r}\boldsymbol{\omega}_{r}\right)\times\textbf{J}_{w}\boldsymbol{\omega}_{w}+\textbf{\textit{u}}+\textbf{\textit{d}} \right.\\
    &\left. \quad+\textbf{J}_{0}\left(\boldsymbol{\omega}_{e}^{\times}\textbf{C}_{r}^{b}\boldsymbol{\omega}_{r}-\textbf{C}_{r}^{b}\boldsymbol{\omega}_{r}\right) \right.\\ 
    &\left.\quad+\delta\textbf{J}\left(\boldsymbol{\omega}_{e}^{\times}\textbf{C}_{r}^{b}\boldsymbol{\omega}_{r}-\textbf{C}_{r}^{b}\boldsymbol{\omega}_{r}\right) \right]\\
        &=\textbf{J}_{0}^{-1}\textbf{\textit{u}}+\textbf{J}_{0}^{-1}\left[-\left(\boldsymbol{\omega}_{e} + \textbf{C}^{b}_{r}\boldsymbol{\omega}_{r}\right)\times\textbf{J}_{0}\left(\boldsymbol{\omega}_{e} + \textbf{C}^{b}_{r}\boldsymbol{\omega}_{r}\right) \right.\\
    &\left. \quad-\left(\boldsymbol{\omega}_{e}+\textbf{C}^{b}_{r}\boldsymbol{\omega}_{r}\right)\times\textbf{J}_{w}\boldsymbol{\omega}_{w}
    +\textbf{J}_{0}\left(\boldsymbol{\omega}_{e}^{\times}\textbf{C}_{r}^{b}\boldsymbol{\omega}_{r}-\textbf{C}_{r}^{b}\boldsymbol{\omega}_{r}\right)\right]\\
    &\quad+\textbf{J}_{0}^{-1}\left[-\left(\boldsymbol{\omega}_{e} + \textbf{C}^{b}_{r}\boldsymbol{\omega}_{r}\right)\times\delta\textbf{J}\left(\boldsymbol{\omega}_{e} + \textbf{C}^{b}_{r}\boldsymbol{\omega}_{r}\right) \right.\\
    &\left. \quad+\textbf{\textit{d}}+\delta\textbf{J}\left(\boldsymbol{\omega}_{e}^{\times}\textbf{C}_{r}^{b}\boldsymbol{\omega}_{r}
    -\textbf{C}_{r}^{b}\boldsymbol{\omega}_{r}\right)\right]\\
    &\quad+\delta\hat{\textbf{J}}\left[-\left(\boldsymbol{\omega}_{e} + \textbf{C}^{b}_{r}\boldsymbol{\omega}_{r}\right)\times\textbf{J}_{0}\left(\boldsymbol{\omega}_{e} + \textbf{C}^{b}_{r}\boldsymbol{\omega}_{r}\right) \right.\\
    &\left.\quad-\left(\boldsymbol{\omega}_{e} + \textbf{C}^{b}_{r}\boldsymbol{\omega}_{r}\right)\times\delta\textbf{J}\left(\boldsymbol{\omega}_{e} + \textbf{C}^{b}_{r}\boldsymbol{\omega}_{r}\right) \right.\\
    &\left.\quad-\left(\boldsymbol{\omega}_{e} + \textbf{C}^{b}_{r}\boldsymbol{\omega}_{r}\right)\times\textbf{J}_{w}\boldsymbol{\omega}_{w}\right.\\   
    &\left.\quad + \textbf{\textit{u}}+\textbf{\textit{d}}+\textbf{J}_{0}\left(\boldsymbol{\omega}_{e}^{\times}\textbf{C}_{r}^{b}\boldsymbol{\omega}_{r}-\textbf{C}_{r}^{b}\boldsymbol{\omega}_{r}\right)\right]\\
    &=\textbf{J}_{0}^{-1}\textbf{\textit{u}}+\overline{\textbf{\textit{d}}}_{1}+\overline{\textbf{\textit{d}}}_{2}
    +\overline{\textbf{\textit{d}}}_{3}
    \end{aligned}
\end{equation}
where \begin{math}\overline{\textbf{\textit{d}}}_{1}\end{math}, \begin{math}\overline{\textbf{\textit{d}}}_{2}\end{math}, and
\begin{math}\overline{\textbf{\textit{d}}}_{3}\end{math} are the second, third, and fourth terms in \eqref{eq:tmp4} that are assumed to be unknown but bounded.

Therefore, \eqref{eq:sAttDot} substituted by \eqref{eq:kmtAttErr} and \eqref{eq:tmp4} is derived as
\begin{equation}
    \begin{aligned}    \dot{\textbf{\textit{s}}}&=\textbf{J}_{0}^{-1}\textbf{\textit{u}}+\overline{\textbf{\textit{d}}}_{1}+\overline{\textbf{\textit{d}}}_{2}+\overline{\textbf{\textit{d}}}_{3}\\
    &+\frac{1}{2}\left(\rho\boldsymbol{\upalpha}\boldsymbol{\upTheta}+\boldsymbol{\upbeta}\textbf{I}_{3\times3}\right)
    \left(q_{e4}\textbf{I}_{3\times3}+\textbf{\textit{q}}_{ev}^{\times}\right)\boldsymbol{\omega}_{e}\\    &=\textbf{J}_{0}^{-1}\textbf{\textit{u}}+\overline{\textbf{\textit{d}}}_{1}+\overline{\textbf{\textit{d}}}_{2}+\overline{\textbf{\textit{d}}}_{3}+\overline{\textbf{\textit{d}}}_{4}\\
    &=\textbf{J}_{0}^{-1}\textbf{\textit{u}}+\lambda_{max}^{-1}\overline{\textbf{\textit{d}}}
    \end{aligned}
\end{equation}
where $\lambda_{max}$ is the maximum eigenvalue of $\textbf{J}_{0}$ and $\lambda_{max}^{-1}\overline{\textbf{\textit{d}}}=overline{\textbf{\textit{d}}}_{1}+\overline{\textbf{\textit{d}}}_{2}+\overline{\textbf{\textit{d}}}_{3}+\overline{\textbf{\textit{d}}}_{4}
$.

As a result, the time derivative of the Lyapunov function in \eqref{eq:VAtt2} leads to 
\begin{equation}
    \begin{aligned}
        \dot{V}&=\textbf{\textit{s}}^{T}\textbf{J}_{0}\dot{\textbf{\textit{s}}}+\frac{1}{\gamma}\left(K-K^{*}\right)\dot{K}\\
        &=\textbf{\textit{s}}^{T}\left(\textbf{\textit{u}} +\lambda_{max}^{-1}\textbf{J}_{0}\overline{\textbf{\textit{d}}}\right)+\frac{1}{\gamma}\left(K-K^{*}\right)\dot{K}\\
        &=\textbf{\textit{s}}^{T}\left(-\frac{K}{\varepsilon}\textbf{\textit{s}}+\lambda_{max}^{-1}\textbf{J}_{0}\overline{\textbf{\textit{d}}}\right)+\frac{1}{\gamma}\left(K-K^{*}\right)\dot{K}\\
        &=\lambda_{max}^{-1}\textbf{\textit{s}}^{T}\textbf{J}_{0}\overline{\textbf{\textit{d}}}-\frac{K}{\varepsilon}\textbf{\textit{s}}^{T}\textbf{\textit{s}}+
        \frac{1}{\gamma}\left(K-K^{*}\right)\dot{K}
    \end{aligned}
\end{equation}

Since \begin{math}\textbf{J}_{0}\end{math} is a symmetric matrix, the inequality equation is derived by the spectral theorem
\begin{equation}
    \textbf{\textit{s}}^{T}\textbf{J}_{0}\overline{\textbf{\textit{d}}}\leq
    \lambda_{max}\left\lVert\textbf{\textit{s}}\right\lVert_{2}
    \left\lVert\overline{\textbf{\textit{d}}}\right\lVert_{2}
\end{equation}
The disturbance is assumed to have an unknown bound as
\begin{equation}
    \left\lVert \lambda_{max}^{-1}\overline{\textbf{\textit{d}}} \right\lVert_{2} = 
    \left\lVert \overline{\textbf{\textit{d}}}_{1}  + \overline{\textbf{\textit{d}}}_{2} + \overline{\textbf{\textit{d}}}_{3} + \overline{\textbf{\textit{d}}}_{3} \right\lVert_{2} < \lambda_{max}^{-1}\overline{D}
\end{equation}
Therefore, it follows that
\begin{equation}
\begin{aligned}
    \dot{V}&\leq\left\lVert \textbf{\textit{s}} \right\lVert_{2}
    \left\lVert\overline{\textbf{\textit{d}}}\right\lVert_{2}-\frac{K}{\varepsilon}\left\lVert\textbf{\textit{s}}\right\lVert_{2}^{2}
    +\frac{1}{\gamma}\left(K-K^{*}\right)\dot{K}\\
    &<\left\lVert\textbf{\textit{s}}\right\lVert_{2}\left( \overline{D} -\frac{K}{\varepsilon} \left\lVert\textbf{\textit{s}}\right\lVert_{2}\right)+\frac{1}{\gamma}(K-K^{*})\dot{K}
\end{aligned}
\end{equation}
where \begin{math}\left\lVert\overline{\textbf{\textit{d}}}\right\lVert_{2}<\overline{D}\end{math}.
In the region \begin{math}\left\lVert\textbf{\textit{s}}\right\lVert_{2}>\varepsilon\end{math}, it leads to
\begin{equation}
\begin{aligned}
    \dot{V}&<\left\lVert\textbf{\textit{s}}\right\lVert_{2}\left(\overline{D}-K\right)+\frac{1}{\gamma}\left(K-K^{*}\right)\dot{K}\\
    &=\left\lVert\textbf{\textit{s}}\right\lVert_{2}\left(\overline{D}-K^{*}\right)+\frac{1}{\gamma}\left(K-K^{*}\right)\dot{K}-\left\lVert\textbf{\textit{s}}\right\lVert_{2}\left(K-K^{*}\right)\\
    &=-\beta_{s}\left\lVert\textbf{\textit{s}}\right\lVert_{2}-\left\lvert K-K^{*}\right\rvert\left(-\left\lVert\textbf{\textit{s}}\right\lVert_{2}+\frac{1}{\gamma}\dot{K}\right)
\end{aligned}
\end{equation}
where \begin{math}\beta_{s}\triangleq K^{*}-\overline{D}>0\end{math} is defined and the equality condition \begin{math}K-K^{*}=-\left\lvert K-K^{*} \right\rvert\end{math} is utilized. Next, a parameter \begin{math}\beta_{K}>0\end{math} is introduced such that
\begin{equation}
\begin{aligned}
    \dot{V}&<-\beta_{s}\left\lVert\textbf{\textit{s}}\right\lVert_{2}-\left\lvert K-K^{*} \right\lvert \left( -\left\lVert \textbf{\textit{s}} \right\lVert_{2} + \frac{1}{\gamma} \dot{K} \right) \\
    &\quad + \beta_{K} \left\lvert K-K^{*} \right\lvert - \beta_{K} \left\lvert K-K^{*} \right\lvert \\
    &= - \beta_{s}\left\lVert\textbf{\textit{s}}\right\lVert_{2} - \left\lvert K-K^{*} \right\lvert \left( - \left\lVert \textbf{\textit{s}} \right\lVert_{2} + \frac{1}{\gamma}\dot{K} - \beta_{K} \right) \\
    &\quad -\beta_{K} \left\lvert K-K^{*} \right\lvert \\
    &= - \beta_{s}\left\lVert\textbf{\textit{s}}\right\lVert_{2} - \beta_{K} \left\lvert K-K^{*} \right\lvert - \kappa
\end{aligned}
\end{equation}
where \begin{math}K-K^{*}=-\left|K-K^{*}\right|\end{math} because \begin{math}K \leq K^{*}
\end{math} and \begin{math}
    \kappa\triangleq\left\lvert K-K^{*} \right\lvert \left( - \left\lVert \textbf{\textit{s}} \right\lVert_{2} + \frac{1}{\gamma}\dot{K}-\beta_{K}\right)
\end{math}.

As a consequence, it is derived as
\begin{equation}\label{eq:tmp12}
    \begin{aligned}
        \dot{V}&< - \beta_{s}\sqrt{2} \cdot \frac{\left\lVert \textbf{\textit{s}} \right\lVert_{2}}{\sqrt{2}} - \beta_{K}\sqrt{2\gamma} \cdot \frac{\left\lvert K-K^{*} \right\lvert}{\sqrt{2\gamma}} - \kappa \\ 
        &\leq - \min{\left( \beta_{s}\sqrt{2}, \beta_{K}\sqrt{2\gamma} \right)} \cdot \left( \frac{\left\lVert \textbf{\textit{s}} \right\lVert_{2}}{\sqrt{2}} + \frac{\left\lvert K-K^{*} \right\lvert}{\sqrt{2\gamma}} \right) - \kappa \\
        &\leq -\beta \cdot V^{1/2} - \kappa
    \end{aligned}
\end{equation}
 where \begin{math}
     \beta\triangleq\sqrt{2}\min{\left( \beta_{s}, \beta_{K}\sqrt{\gamma} \right)}>0
 \end{math}.
 
 The positiveness of parameter \begin{math}\kappa\end{math} suggests that the \begin{math}\gamma\end{math} must be selected properly. If \begin{math}\kappa>0\end{math}, it yields
 \begin{equation}
    \gamma_{att} < \frac{\dot{K}}{\beta_{K} + \left\lVert \textbf{\textit{s}} \right\lVert_{2}}
 \end{equation}
 Let us define $\Gamma_{att} (t) = \frac{\dot{K}}{\beta_{K} + \left\lVert \textbf{\textit{s}} \right\lVert_{2}}$ and the substitution of the adaptive law \eqref{eq:adtLaw} derives
\begin{equation}
\begin{aligned}
    \Gamma_{att}(t) &= \eta \cdot \frac{K(t) \left( \frac{\left\lVert \textbf{\textit{s}} \right\lVert_{2}}{\varepsilon} - 1 \right) + K_{0}}{\beta_{K} + \left\lVert \textbf{\textit{s}} \right\lVert_{2}} \\
    &\geq \eta \cdot \frac{K_{0} \cdot \frac{\left\lVert \textbf{\textit{s}} \right\lVert_{2}}{\varepsilon} }{\beta_{K} + \left\lVert \textbf{\textit{s}} \right\lVert_{2}} \\
    &= \eta \cdot \frac{ \frac{K_{0}}{\varepsilon} }{ \frac{\beta_{K}}{\left\lVert \textbf{\textit{s}} \right\lVert_{2}} + 1 } \\
\end{aligned}
\end{equation}
where $K(t) \geq K_{0}$ is utilized. 

In the region \begin{math}\left\lVert\textbf{\textit{s}}\right\lVert_{2}>\varepsilon\end{math}, $\Gamma_{att}(t)$ is derived as
 \begin{equation}
     \Gamma_{att}(t) > \eta \frac{ \frac{K_{0}}{\varepsilon} }{ \frac{\beta_{K}}{\varepsilon} + 1 } = \eta \frac{ K_{0} }{ \beta_{K} + \varepsilon}
 \end{equation}

 Therefore, parameter \begin{math}\gamma_{att}\end{math} must be chosen to satisfy
 \begin{equation}
     \gamma_{att} < \frac{\eta K_{0}}{\beta_{K} + \varepsilon}.
 \end{equation}

Equation \eqref{eq:tmp12} shows that the Lyapunov function satisfies \begin{math}
    \dot{V} \leq -\beta\cdot V^{1/2} - \kappa < -\beta\cdot V^{1/2}
\end{math}. Hence, the finite-time convergence of the sliding variable into the region \begin{math}
    \left\lVert \textbf{\textit{s}} \right\lVert_{2} \leq \varepsilon
\end{math} is guaranteed by Lemma \ref{lemVdotV}. 

\end{proof}

\begin{lemma}\label{lemBoundOnQe}
If the $i$th component of the attitude sliding variable is bounded by the region $\left\lvert s_{i} \left(t\right) \right\lvert \leq \varepsilon_{i}$, then the quaternion error is bounded as $\left\lvert q_{ei} \left(t\right) \right\lvert \leq Q_{ei}$ for $i=1,2,3$. $Q_{ei}$ is the unique solution of the equation $F(Q_{ei})=0$, where the function $F(x)$ is defined in \eqref{eq:tmp19}. \textcolor{black}{Similarly to Lemma \ref{lemBoundOnRe}, this can also be applied to derive control parameters that have physical significance and adhere to the control requirement.}
\end{lemma}

\begin{proof}
    Let the Lyapunov function be 
    \begin{equation}
        V=\textbf{\textit{q}}_{ev}^{T}\textbf{\textit{q}}_{ev} + (1-q_{e4})^2,
    \end{equation}
    then its time derivative is obtained as 
    \begin{equation}
    \begin{aligned}
        \dot{V}&=\textbf{\textit{q}}_{ev}^{T} \boldsymbol{\omega}_{e}\\
        &=\sum_{i=1}^{3}{\dot{V}_{i}}
    \end{aligned}
    \end{equation}
    where $\dot{V}_{i}=q_{ei}\omega_{ei}$.
    The bound on the sliding variable $\left\lvert s_{i} \right\lvert \leq \varepsilon_{i}$ leads to the inequality equation as
    \begin{equation}
        G(q_{ei}) \leq \omega_{ei} \leq F(q_{ei})
    \end{equation}
    where $F(x)$ and $G(x)$ are the functions given in \eqref{eq:tmp19} and \eqref{eq:tmp20}.
    Since $F(x)$ and $G(x)$ are the same with \eqref{eq:tmp19} and \eqref{eq:tmp20} of Lemma \ref{lemBoundOnRe}, the uniqueness of the solution for the equations $F(Q_{ei,p})=0$ and $G(Q_{ei,n})=0$ and their relationship $Q_{ei,p} = -Q_{ei,n} = Q_{ei}$ are proven in the same way. Therefore, the region of attraction $\mathbb{L}_{2} = \{q_{ei} \in \mathbb{R}\, |\, \left| q_{ei} \right| \leq Q_{ei}\}$ derives that $\dot{V}_{i}$ and $\dot{V}$ are negative when $\left\lvert q_{ei} \right\lvert > Q_{ei}$.
    
\end{proof}

\section{Numerical Simulation}
\label{s:5sim}

\textcolor{black}{\subsection{Simulation Setting}}
The orbital perturbations induce the external disturbance force in space. The Earth's geopotential model JGM3 determines the aspherical gravity by degrees and orders of 70. The third-body gravity from the Sun and Moon that treated as point masses is applied. The gravity gradient torque from the point-mass Earth's gravity exerts a disturbance on the attitude control system. The residual dipole moments of $10^{-3}$ $\mathrm{Am^2}$ in a random direction and the Earth's magnetic field of international geomagnetic reference field model generates magnetic disturbance torque. The atmospheric drag and solar radiation pressure are the common disturbances for the relative orbit control and attitude systems. The atmospheric density is calculated from the exponential model, and the solar radiation pressure adopts a constant solar flux. The internal uncertainties on the deputy satellite originate from its mass and moments of inertia. The non-Gaussian random errors demonstrate these uncertainties for the software simulation.

\textcolor{black}{\subsection{Control System's Block Diagram}}
The ROCS and ACS \textcolor{black}{are developed based on separate dynamics as in \eqref{eq:dynOrb}, \eqref{eq:dynOrb_A1A2Eb}, \eqref{eq:dynAtt} and each is a complete control plant, respectively. In \figurename{\ref{fig:FSSDiagram}}, the blocks in the upper and lower dashed blue boxes constitute ROCS and ACS. However, two control plants should operate in a larger closed loop inter-dependently to organize an entire formation flying system.} \figurename{\ref{fig:FSSDiagram}} depicts the block diagram of both systems. The software algorithm of the ROCS, the orbit controller box, provides the relative orbit error states as inputs and generates the reference control input $\textbf{\textit{u}}_{orb,r}$. Let the reference control input from the control algorithm be denoted by $\textbf{\textit{u}}_{orb,0}^{l}$ that corresponds to $\textbf{\textit{u}}_{orb}$ in the equation. The superscript indicates the variable's frame as in Section~\ref{s:2mission}, and the subscript $0$ is a temporary number of the process. The zero-order hold function holds the reference control input constant for the sampling interval $T_{s,orb}$. The sampling interval for the ROCS system should be properly selected for both attitude control and immediate update of control input. $T_{s,orb}$ of 5 seconds is used in this simulation. The discretized signal $\textbf{\textit{u}}_{orb,1}^{l}$ is then converted into the body frame as $\textbf{\textit{u}}_{orb,1}^{b}$ to give an input to the orbital actuator and the propulsion system. 

The resolution and the saturation limit of thrust firing and the restricted degrees of freedom contribute to the thrust error and produce $\textbf{\textit{u}}_{orb,2}^{b}$. In the numerical simulation, the fired thrust $\textbf{\textit{u}}_{orb,2}^{b}$ is converted into the $\textbf{\textit{u}}_{orb,2}^{l}$ and $\textbf{\textit{u}}_{orb,2}^{i}$, sequentially. Since the control algorithm is formulated in the LVLH frame, $\textbf{\textit{u}}_{orb}$ in the simulation result graphs is the same with $\textbf{\textit{u}}_{orb,2}^{l}$. $\textbf{\textit{u}}_{orb,2}^{i}$ is the control input in the inertial frame and integrated with the orbital perturbations and mass uncertainties. The numerical integrator, ode45 of MATLAB provides the orbital states with every time step \textcolor{black}{($\Delta t$)} of 0.01 second.

The ACS generates the reference attitude from the zero-ordered thrust vector in the body frame $\textbf{\textit{u}}_{orb,1}^{b}$ that is transmitted to ACS every orbital sampling time. Since the reference attitude possibly alters every orbital sampling time, the ACS should be able to confine the sliding variable within $T_{s,orb}$. Therefore, the orbital sampling time is critical for the attitude control performance. The attitude control problem of the time limit $T_{s,orb}$ repeats for the whole simulation duration. The reference torque is integrated with the disturbance torques and the moment of inertia uncertainties using the attitude dynamics with reaction wheels. The numerical integrator and the time step for the ACS system are the same as the ones for the ROCS. 

\begin{figure*}
\begin{center}
\includegraphics[width=1.0\linewidth]{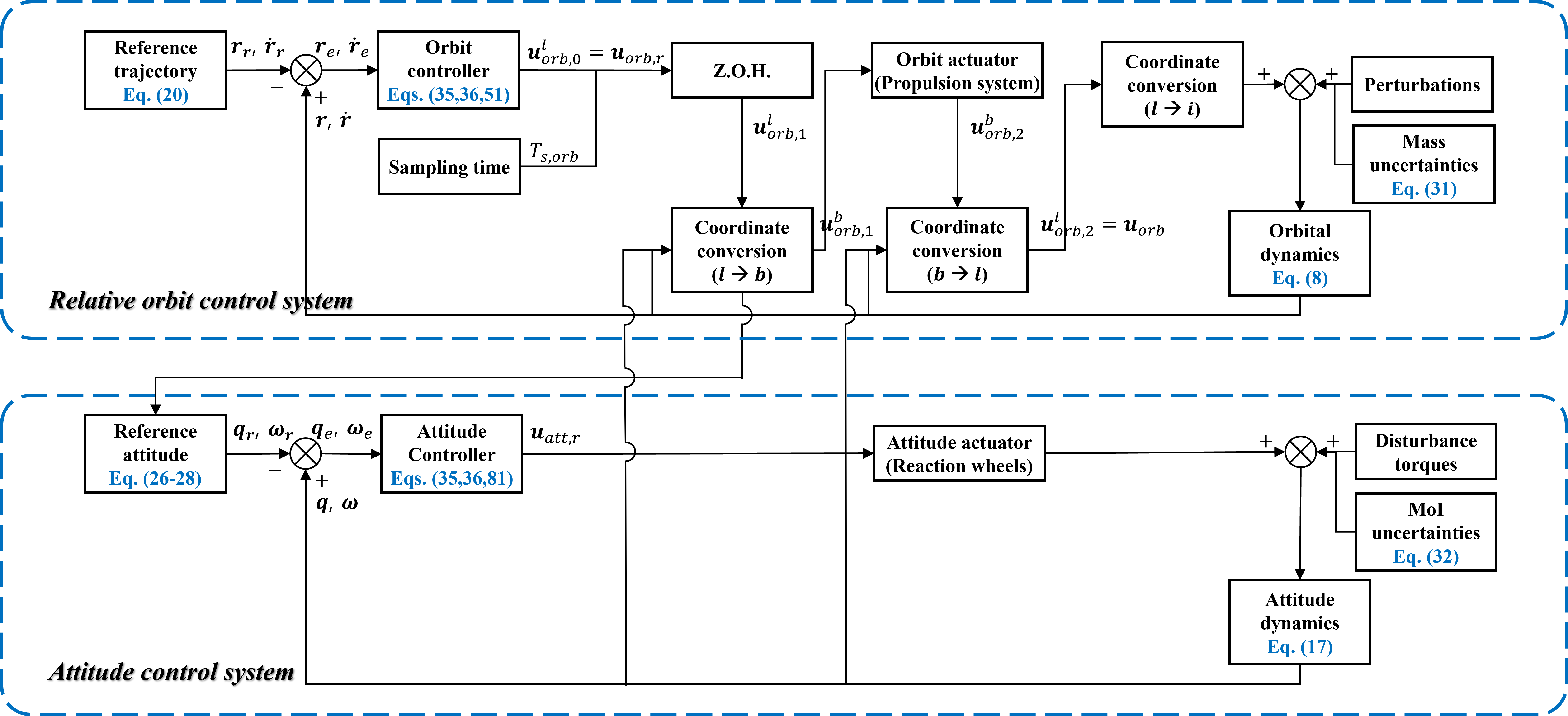}
\end{center}
\caption{Block diagram of control loops in formation flying system}
\label{fig:FSSDiagram}
\end{figure*}

\textcolor{black}{\subsection{Simulation Parameters}}

The procedure to search for the proper control parameters is as follows. The parameters that constitute the sliding variable are $\rho$, $\boldsymbol{\upalpha}$, and $\boldsymbol{\upbeta}$. \textcolor{black}{According to the sliding variable design in Eqs. \eqref{eq:sOrb} and \eqref{eq:sAtt}, the absence of the $\textnormal{sig}^{\rho}$ term implies a conventional linear sliding variable.} The parameter $\rho$ determines the convergence rate of the nonsingular fast terminal sliding variable when it retains zero. It can be readily shown that if $\left\lvert x \right\lvert < 1$ for an arbitrary variable $x$, $x$ achieves zero much faster as $\rho$ is smaller. On the contrary, when $1 < \left\lvert x \right\lvert$, larger $\rho$ shows faster convergence. Since the initial relative position error is likely to be larger than one meter and the definition of the quaternion restricts the initial quaternion error below one, it is recommended that $\rho_{orb}$ and $\rho_{att}$ are selected as 1.90 and 1.10, respectively. 

\textcolor{black}{
Lemma \ref{lemBoundOnRe} analytically determines $\boldsymbol{\upalpha}_{orb}$, $\boldsymbol{\upbeta}_{orb}$, and $\varepsilon_{orb}$ for ROCS and Lemma \ref{lemBoundOnQe} achieves the parameters for the ACS in the same manner. Lemma \ref{lemBoundOnRe} suggests that $R_{ei}$ is the unique solution of the function $F(x)=0$ and can be reorganized as 
\begin{equation}\label{eq:tmp29}
    \frac{F(R_{ei})}{\varepsilon_{i}} = 1 - \frac{\alpha_{i}}{\varepsilon_{i}}\textnormal{sig}^{\rho}R_{ei} - \frac{\beta_{i}}{\varepsilon_{i}}R_{ei} = 0
\end{equation}
Assuming that all three components of $\alpha_{i}$ and $\beta_{i}$ are chosen to be the same, the requirement in \eqref{eq:reqROCS} derives the inequality condition of $R_{ei}$ as
\begin{equation}
    R_{ei} \leq \sqrt{3}
\end{equation}
Assuming the largest $R_{ei}$, it leads \eqref{eq:tmp29} to be
\begin{equation}\label{eq:tmp30}
    1 - \frac{\alpha_{i}}{\varepsilon_{i}}\textnormal{sig}^{\rho}\sqrt{3} - \frac{\beta_{i}}{\varepsilon_{i}}\sqrt{3} = 0
\end{equation}
As a consequence, $\frac{\alpha_{i}}{\varepsilon_{i}}$ and $\frac{\beta_{i}}{\varepsilon_{i}}$ for ROCS should satisfy the inequality conditions as follows:
\begin{equation}\label{eq:tmp31}
    \frac{\alpha_{orb,i}}{\varepsilon_{orb,i}} \geq 0.29
\end{equation}
\begin{equation}\label{eq:tmp31_2}
   \frac{\beta_{orb,i}}{\varepsilon_{orb,i}} \geq 0.29
\end{equation}
In a similar manner, the requirement \eqref{eq:reqACS} derives $Q_{ei}$ in Lemma \ref{lemBoundOnQe} as 
\begin{equation}
    Q_{ei} \leq \sin{ \left( \frac{\sqrt{3}}{2} \cdot \frac{\pi}{180} \right) } = 1.51 \times 10^{-2}
\end{equation}
Following the same procedure as in \eqref{eq:tmp29}--\eqref{eq:tmp31_2}, the inequality conditions of $\frac{\alpha_{i}}{\varepsilon_{i}}$ and $\frac{\beta_{i}}{\varepsilon_{i}}$ for ACS are written as follows:
\begin{equation}\label{eq:tmp32}
    \frac{\alpha_{att,i}}{\varepsilon_{att,i}} \geq 24.0
\end{equation}
\begin{equation}\label{eq:tmp32_2}
    \frac{\beta_{att,i}}{\varepsilon_{att,i}} \geq 24.0
\end{equation}
}

Equations \eqref{eq:tmp31}--\eqref{eq:tmp32_2} suggest that the ratios of $\frac{\alpha_{i}}{\varepsilon_{i}}$ and $\frac{\beta_{i}}{\varepsilon_{i}}$ give the critical importance in the error bounds and $\varepsilon_{i}$ is not solely determined. In addition, the parameters $\boldsymbol{\upalpha}$ and $\boldsymbol{\upbeta}$ attach weights to the error state relative to its time derivatives. As $\boldsymbol{\upalpha}$ and $\boldsymbol{\upbeta}$ become sufficiently large, it mitigates the direct influence from the control input on the relative velocity or the angular velocity, and the controlled states suffer less from sudden rises. Several simulations and intuition help with the empirical tuning of those values.

\textcolor{black}{
\begin{table}[b]
\renewcommand{\arraystretch}{1.3}
\caption{\textcolor{black}{Simulation parameters and physical meaning}}
\label{tab:simParamPhys2}
\centering
\begin{tabular}{c|c}
\hline\hline
\textcolor{black}{Parameter} & \textcolor{black}{Physical meaning}  \\
\hline
\textcolor{black}{$\rho$} & \textcolor{black}{Convergence rate of the sliding variable}  \\
\hline
\textcolor{black}{$\alpha_{i} / \varepsilon$, $\beta_{i} / \varepsilon$} & \textcolor{black}{Pre-designated error bound}  \\
\hline
\textcolor{black}{$K\left( 0 \right)$} & \textcolor{black}{\makecell{Initial control gain that \\ should be chosen to be less than the actuator limit}}  \\
\hline
\textcolor{black}{$K_{0}$} & \textcolor{black}{Lower limit of the control gain $K\left( t \right)$}  \\
\hline
\textcolor{black}{$\eta$} & \textcolor{black}{\makecell{Proportional constant that the rate of change in $K\left( t \right)$ and \\ should satisfy $0 < \eta \cdot \Delta t \leq 1$}}   \\
\hline\hline
\end{tabular}
\end{table}}


The adaptive control law has the advantage that only a few parameters are to be selected. Suppose that the sliding variable approaches the vicinity of $\varepsilon$ to have the same order of $\varepsilon$. Then the control gain $K(t)$ has the same order as the control input $\textbf{\textit{u}}(t)$. Since the error is disposed to decrease, the initial control gain $K(0)$ could be chosen considering the actuator limit in Table \ref{tab:actuator}. The reorganized formula of the adaptive law in \eqref{eq:adtLaw2} suggests that the relative magnitude of $\left\lVert \textbf{\textit{s}} \right\lVert_{2} / \varepsilon$ and $\left( 1-\frac{K_{0}}{K(t)} \right)$ determines the sign of $\dot{K}$. 
\begin{equation}\label{eq:adtLaw2}
    \dot{K}(t)=\eta K(t) \left(\left\lVert-\frac{\textbf{\textit{s}}}{\varepsilon} \right\rVert_{2}-\left(1-\frac{K_{0}}{K(t)} \right)\right)
\end{equation}
Therefore, the parameters $K_{0}$ and $\eta$ affect how sensitively the sign and magnitude of $\dot{K}$ change. In the same manner as other parameters, it needs a number of simulations and intuition. 

\begin{table}[b]
\renewcommand{\arraystretch}{1.3}
\caption{Actuator specification of attitude and orbit control systems}
\label{tab:actuator}
\centering
\begin{tabular}{c|c|c}
\hline\hline
Actuator & \multicolumn{2}{c}{Specification} \\
\hline
\multirow{2}{*}{Thruster} & \makecell{Maximum thrust\\($u_{orb,i}^{*}$ for $i=x,y,z$)} & $1.00$ (mN) \\
\cline{2-3} & Thrust resolution & $0.01$ (mN) \\
\hline
Reaction wheels & \makecell{Maximum torque\\($u_{att,i}^{*}$ for $i=1,2,3$)} & $0.23$ (mNm) \\
\hline\hline
\end{tabular}
\end{table}


Since each of $\alpha_{att,i}$ for $i = 1,2,3$ is set to the same value, let $\alpha_{att}=\alpha_{att,i=1,2,3}$ be defined. In a similar manner $\beta_{att}$, $\alpha_{orb}$, and $\beta_{orb}$ are defined. If $\alpha_{att}$, $\beta_{att}$, $\rho_{att}$, $\alpha_{orb}$, $\beta_{orb}$, and $\rho_{orb}$ are set as $1.44$, $1.44$, $1.10$, $7.00\times10^{-3}$, $7.00\times10^{-3}$, and $1.90$ for $i = 1,2,3$, respectively, then the error and its time derivative have the same order. $\varepsilon_{att}$ and $\varepsilon_{orb}$ are respectively set as $6.00\times10^{-2}$ and $1.20\times10^{-2}$, considering the requirements for the error tolerance and the control input saturation. The control parameters $\eta_{att}$, $K_{att}(0)$, $K_{att,0}$, \textcolor{black}{$\eta_{orb}$}, $K_{orb}(0)$, and $K_{orb,0}$ are selected as $5.00\times10^{-2}$, $2.00\times10^{-4}$, $2.00\times10^{-7}$, \textcolor{black}{$4.50\times10^{-4}$}, $3.00\times10^{-5}$, and $1.00\times10^{-8}$. These parameters are summarized in Table \ref{tab:simParams}.

\begin{table}[b]
\renewcommand{\arraystretch}{1.3}
\caption{Simulation parameters for attitude and orbit control systems}
\label{tab:simParams}
\centering
\begin{tabular}{c|c|c|c}
\hline\hline
$\rho_{att}$ & $\alpha_{att}$ & $\beta_{att}$ & $ $ \\
\hline
$1.10$ & $1.44$ & $1.44$ & $ $ \\
\hline
$\varepsilon_{att}$ & $K_{att}(0)$ & $K_{att,0}$ & $\eta_{att}$ \\
\hline
$6.00\times10^{-2}$ & $2.00\times10^{-4}$ & $2.00\times10^{-7}$ & $5.00\times10^{-2}$ \\
\hline\hline
$\rho_{orb}$ & $\alpha_{orb}$ & $\beta_{orb}$ & $ $ \\
\hline
$1.90$ & $7.00\times10^{-3}$ & $7.00\times10^{-3}$ & $ $ \\
\hline
$\varepsilon_{orb}$ & $K_{orb}(0)$ & $K_{orb,0}$ & \textcolor{black}{$\eta_{orb}$}  \\
\hline
$1.20\times10^{-2}$ & $3.00\times10^{-5}$ & $1.00\times10^{-8}$ & \textcolor{black}{$4.50\times10^{-4}$}  \\
\hline\hline
\end{tabular}
\end{table}

\textcolor{black}{\subsection{Simulation Result}}

\textcolor{black}{According to the top panel of \figurename{\ref{fig:simAlignOrbS}}, the two-norm of the orbital sliding variable rapidly decreases and initially converges below $\varepsilon_{orb}$ at 205 seconds. Afterward, the sliding variable exceeds $\varepsilon_{orb}$ and is suppressed again. The time intervals that $\left\lVert \textbf{\textit{s}}_{orb} \right\lVert_{2}$ grows larger than $\varepsilon_{orb}$ are 345--350, 685--770 $\left(\Delta t_{s_{orb},1}\right)$, 920--1235 $\left( \Delta t_{s_{orb},2} \right)$, 2410--2705 $\left( \Delta t_{s_{orb},3} \right)$, 2755--2795 $\left( \Delta t_{s_{orb},4} \right)$, 2800--2810, 3325--3400 $\left( \Delta t_{s_{orb},5} \right)$, and 3405--3410 seconds. Since the sampling interval is 5 seconds, temporary spikes such as 345--350, 2800--2810, and 3405--3410 seconds last only for a few sampling intervals, and can be regarded as included in the previous intervals.}

\textcolor{black}{The appropriate choice of $\boldsymbol{\upalpha}_{orb}$ and $\boldsymbol{\upbeta}_{orb}$ holds the two-norm of the position error below 3 meters even during $\Delta t_{s_{orb},2}$ and $\Delta t_{s_{orb},3}$. The two-norm of the position error plotted in the middle panel of \figurename{\ref{fig:simAlignOrbS}} initially reduces below the requirement at 45.60 seconds, and then remains in the region thereafter.}

The position error in each direction of the LVLH frame is described at the bottom panel. \textcolor{black}{The cross-track position error mostly remains near zero because the perturbation forces are the smallest in the cross-track direction. The radial and along-track position errors increase due to the disturbances during 750--1200 seconds. If the errors and the sliding variable are accumulated enough to outgrow the thrust's threshold, then the thrust is fired to suppress the errors and the sliding variable. This is how the control law intends, so this phenomenon is repeated to control the errors after 2500 seconds.}

\begin{figure}[b!]
\begin{center}
\includegraphics[width=1.0\linewidth]{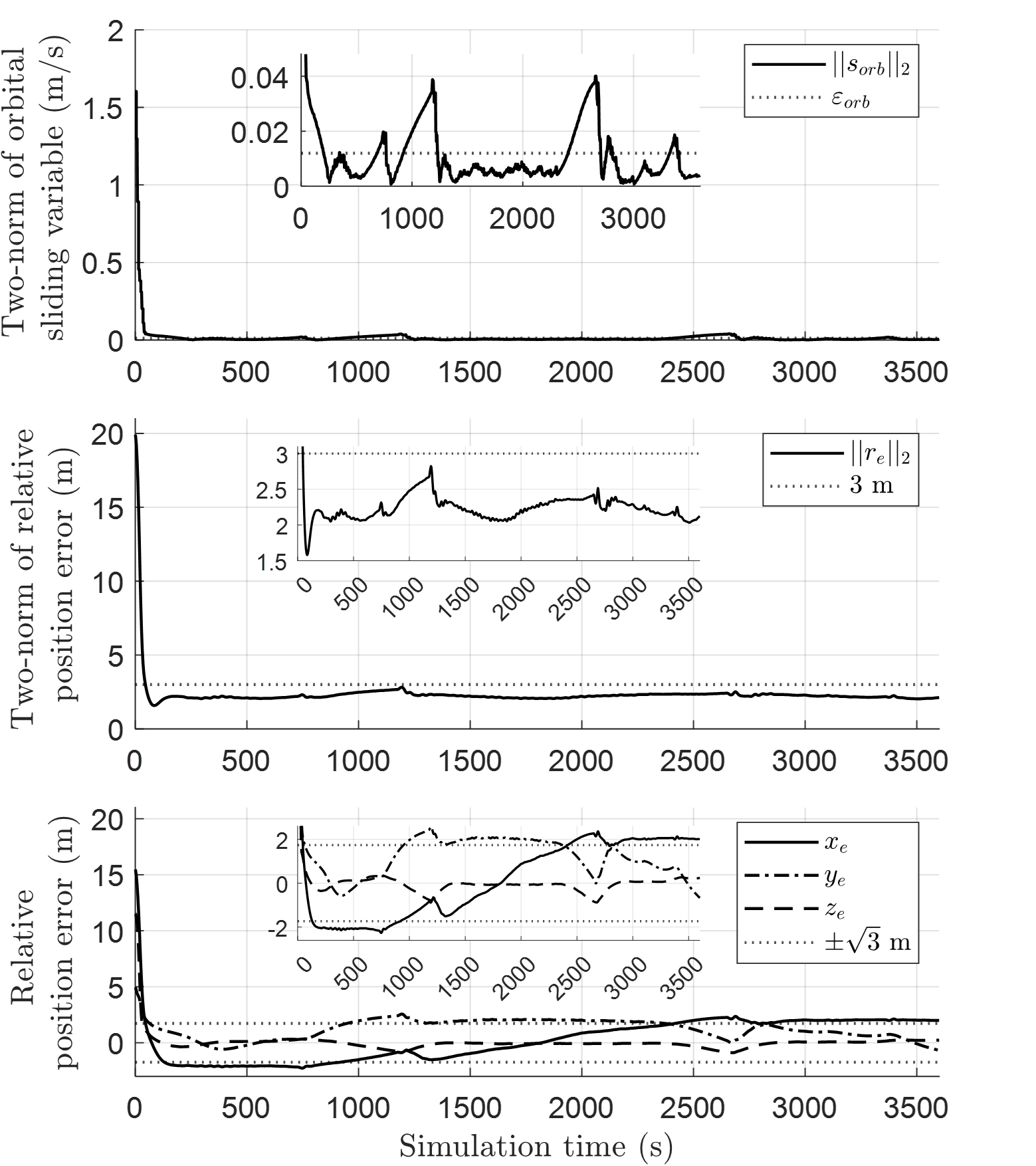}
\end{center}
\caption{Two-norm of the orbital sliding variable and position error}
\label{fig:simAlignOrbS}
\end{figure}


The difference between the reference thrust from the control algorithm $\textbf{\textit{u}}_{orb,r}$ and the fired thrust $\textbf{\textit{u}}_{orb}$ defines the thrust error $\textbf{\textit{u}}_{orb,e}$ as illustrated in \figurename{\ref{fig:simAlignUOrbK}}, where $\textbf{\textit{u}}_{orb,e}=\textbf{\textit{u}}_{orb}-\textbf{\textit{u}}_{orb,r}$. The fired thrust and the thrust error are depicted in the top three panels of \figurename{\ref{fig:simAlignUOrbK}}. The radial thrust error in the first panel initially has a positive value and rapidly decreases during the first 50 seconds. The positive thrust error suggests that the reference thrust is negative and controls the positive orbital error. It also explains the positive along-track and cross-track thrust errors in the second and third panels. After 50 seconds, the magnitude of thrust is almost zero and induces the offset of the radial position error. After the convergence of the sliding variable, only a small magnitude of thrust is fired discontinuously. Because the propulsion system has a thrust resolution of 0.01 mN, the thrust could be fired once the reference thrust is large enough. Compared to the radial and along-track cases, the cross-track thrust easily controls the position error and consumes less fuel after the convergence. \textcolor{black}{The thrust is generated in all three axes to counteract the drastic rise of the sliding variable.} The fourth and fifth panels in \figurename{\ref{fig:simAlignUOrbK}} depict $K$ and $\dot{K}$. According to the adaptation law in \eqref{eq:adtLaw2}, the negative sign of $\dot{K}$ and the converged $K$ indicate that the sliding variable is bounded. These simulation results confirm that the sliding variable is bounded \textcolor{black}{after 350 seconds, except $\Delta t_{s_{orb},\{1,2,...,5\}}$.}

\begin{figure}[t]
\begin{center}
\includegraphics[width=1.0\linewidth]{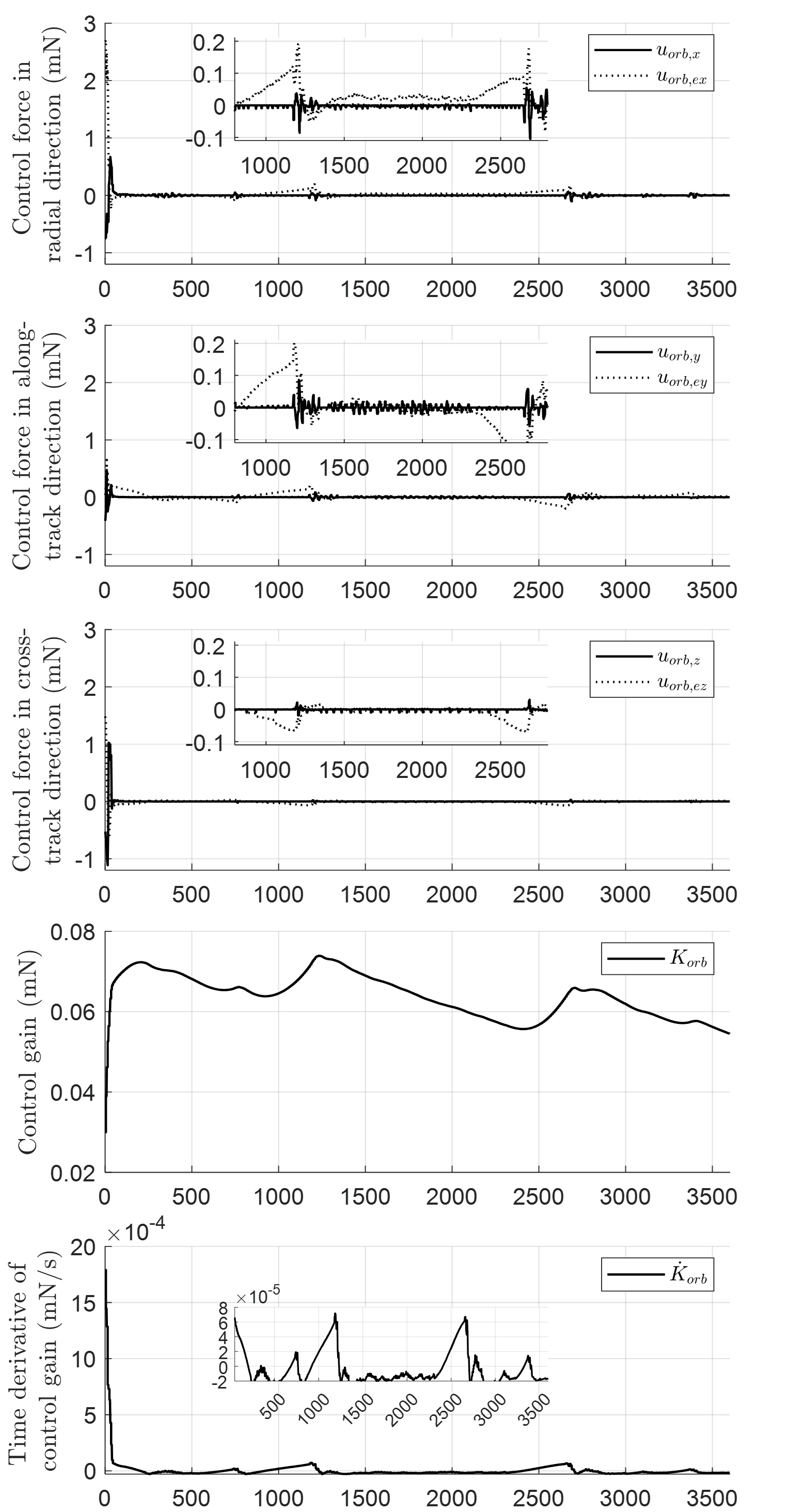}
\end{center}
\caption{Fired thrust, thrust error, orbital control gain and its time derivative}
\label{fig:simAlignUOrbK}
\end{figure}

The abrupt rises in the sliding variable and quaternion errors occur when the reference quaternion changes. According to the simulation result shown in \figurename{\ref{fig:simAlignAttS}}, the two reference quaternions switch \textcolor{black}{117 times for 3600 seconds.} The expectable reference quaternions make it possible to predict the convergence speed of the attitude sliding variable. The initial attitude is usually oriented far from the reference attitude in the beginning. The top two panels in \figurename{\ref{fig:simAlignAttS}} show that the change in the reference attitude could instantly reduce the control error during the first 50 seconds. The ROCS generates the reference thrust more frequently \textcolor{black}{when the orbital sliding variable exceeds during $\Delta t_{s_{orb},\{1,2,...,5\}}$.} The attitude maneuver is very regular and the convergence speed of the attitude sliding variable should be analyzed \textcolor{black}{except $\Delta t_{s_{orb},\{1,2,...,5\}}$.} \textcolor{black}{For example, the attitude sliding variable is bounded by $\left\lVert \textbf{\textit{s}} \right\lVert_{2} \leq \varepsilon_{att}$ within 8.55 seconds on average after 350 seconds. The attitude maneuver is regular, the convergence speed is predicted to be similar for other time intervals.}

As it is described in Section~\ref{s:2mission}, the reference attitude rotates $\pm 10$ deg along the $+\hat{\textbf{\textit{b}}}_{x}$ axis to achieve the controllability in $+\hat{\textbf{\textit{b}}}_{z}$ axis. Therefore, only the $q_{e1}$ in the second panel momentarily increases except for the initial 50 seconds. Since the three components of the quaternion vector are coupled by the cross term in \eqref{eq:qerrAtt}, the increase in $q_{e1}$ causes $q_{e2}$ and $q_{e3}$ to fluctuate (the third and fourth panels).

\begin{figure}[t]
\begin{center}
\includegraphics[width=1.0\linewidth]{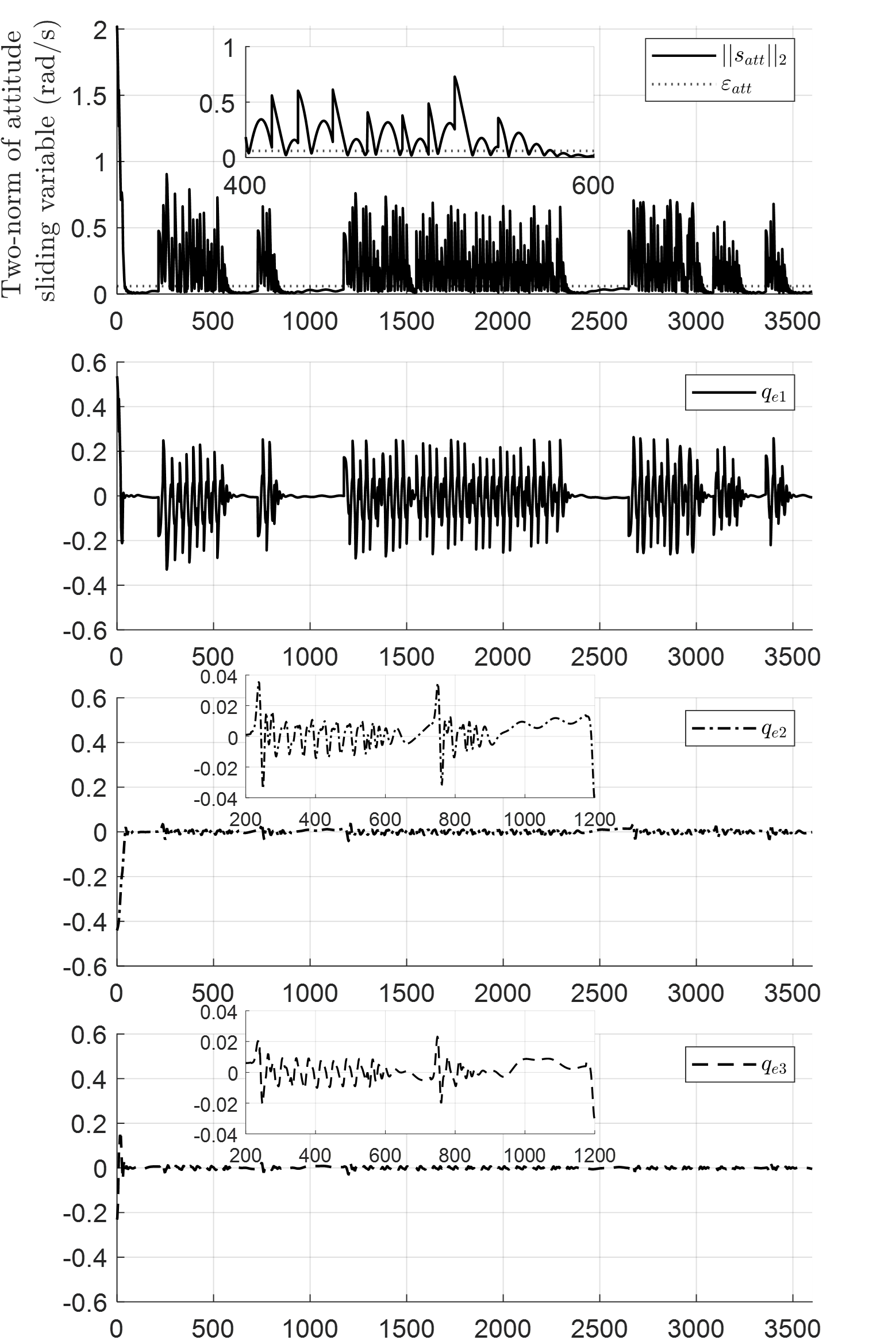}
\end{center}
\caption{Two-norm of the attitude sliding variable and quaternion error}
\label{fig:simAlignAttS}
\end{figure}

The control torque profile in the top three panels in \figurename{\ref{fig:simAlignUAttK}} corresponds to the trend of the quaternion errors. $ u_{att,1} $ is temporarily saturated for about 10 seconds when the $ q_{e1}$ is at the peak. $ u_{att,2} $ and $ u_{att,3} $ are relatively small but never zero unless the three components in quaternion error vectors become zero simultaneously. The fourth and fifth panels explain the rationale that the last convergence speed of the attitude sliding variable is latent. \textcolor{black}{When the attitude sliding variable stayed in the region within $\varepsilon_{att}$ exceeds the bound, $\dot{K}$ becomes positive and $K$ starts to increase.} (top panel of \figurename{\ref{fig:simAlignAttS}}). Theoretically, $\dot{K}$ is negative when $ \left\lVert \textbf{\textit{s}} \right\lVert_{2} \leq \left( 1-\frac{K_{0}}{K} \right) \cdot \varepsilon_{att}$. The simulation verifies that the time that $\dot{K}$ is negative and the time that $ \left\lVert \textbf{\textit{s}} \right\lVert_{2} \leq \left( 1-\frac{K_{0}}{K} \right) \cdot \varepsilon_{att}$ are different by only 0.12 second. Hence, $\dot{K}$ is negative for a considerable time duration and it provokes the significant reduction of $K$ in the fourth panel. 

\begin{figure}[t]
\begin{center}
\includegraphics[width=1.0\linewidth]{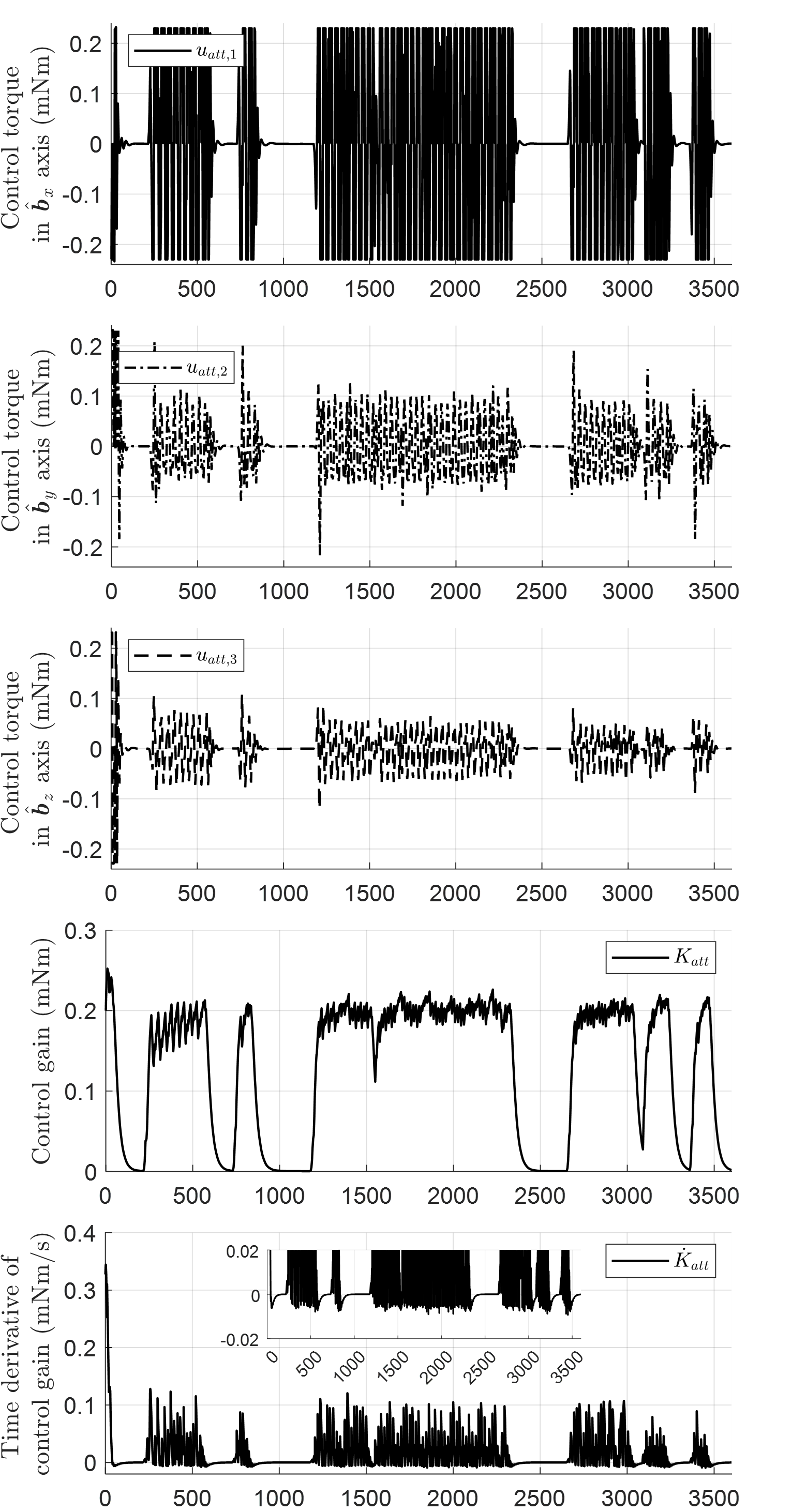}
\end{center}
\caption{Control torque, attitude control gain and its time derivative}
\label{fig:simAlignUAttK}
\end{figure}

\textcolor{black}{The operation time for the alignment in this mission is designed to be constrained within 20 minutes. The simulation time is extended to be 3600 seconds to clarify the control performance. The control performance is compared to the typical approachs such as PD and LQR, where the LQR control algorithms for ROCS and ACS refer to \cite{KTAlfriend} and \cite{LQRAtt}.} According to Table \ref{tab:simCompare}, this NFTSM control algorithm satisfies the mission requirement for \textcolor{black}{59.24 minutes and consumes $7.90\times 10^{-2}$ Ns of fuel.} The weight matrices for the LQR controllers are \textcolor{black}{$\textbf{Q}_{LQR,orb} = diag\{ 6.39\times 10^{-1}, 3.83 \times 10^{-1}, 5.27 \times 10^{-1}, 1.25 \times 10^{-3}, 1.25 \times 10^{-3}, 1.25 \times 10^{-3} \}$, $\textbf{R}_{LQR,orb} = diag \{3.47\times10^{5}, 3.47\times10^{5}, 3.47\times10^{5} \}$, $\textbf{Q}_{LQR,att} = diag\{ 3.00\times10^{-7}, 3.00\times10^{-7}, 3.00\times10^{-7}, 3.00 \times 10^{-4}, 3.00 \times 10^{-4}, 3.00 \times 10^{-4} \}$ and $\textbf{R}_{LQR,att} = diag \{1.16\times10^{4}, 1.16\times10^{4}, 1.16\times10^{4} \}$.} The LQR algorithms retain the satellite within the requirement only for \textcolor{black}{5.50} minutes while it consumes \textcolor{black}{7.02} Ns, \textcolor{black}{which means it drives the spacecraft only near the reference trajectory consuming the maximum thrusts.} \textcolor{black}{On the other hand, the PD controllers have gains as $K_{P,orb} = 6.39 \times 10^{-5}$, $K_{D,orb} = 1.45 \times 10^{-3}$, $K_{P,att} = 2.00 \times 10^{-4}$, and $K_{D,att} = 3.00 \times 10^{-3}$. This approach achieves the aligned time of 16.48 minutes while consuming $8.17 \times 10^{-1}$ Ns.} Therefore, \textcolor{black}{compared to PD and LQR controllers, the NFTSM control algorithm extends the aligned time and improves the fuel efficiency.}

\begin{table}[t]
\renewcommand{\arraystretch}{1.3}
\caption{Comparison of control performance}
\label{tab:simCompare}
\centering
\begin{tabular}{c|c|c|c}
\hline\hline
\textbf{ROCS} & \textbf{ACS} & \textbf{Aligned time} & \textbf{Fuel consumption} \\
\hline
\textcolor{black}{PD} & \textcolor{black}{PD} & \textcolor{black}{$16.48$ min} & \textcolor{black}{$8.17 \times 10^{-1}$ Ns} \\
\hline
\textcolor{black}{LQR} & \textcolor{black}{LQR} & \textcolor{black}{$5.50$ min} & \textcolor{black}{$7.02$ Ns} \\
\hline
NFTSM & NFTSM & \textcolor{black}{$59.24$ min} & \textcolor{black}{$7.90 \times 10^{-2}$ Ns} \\
\hline\hline
\end{tabular}
\end{table}

\textcolor{black}{Two indices are introduced to compare the control performance - aligned time and fuel consumption. Since the primary objective of the mission is to 'align 'two satellites, the system performance is defined from the perspective of orbit control. The aligned time means the total duration that the satellite satisfies \eqref{eq:reqROCS}. The fuel consumption is calculated as the total sum of control input multiplied by the time step. The longer the aligned time is and the less fuel consumption, the better the control performance is. The two-norm of the relative position error and the control force are compared and illustrated in \figurename{\ref{fig:simAlignComp}}. The top panel demonstrates that the LQR controller barely manages to suppress the position error near the requirement. While the PD controller outperforms the LQR regarding error reduction, the error converges after quite a long time. In contrast, the NFTSM successfully restrains the error. The bottom panel compares the magnitude of the control force. The LQR controller consistently generates almost the maximum magnitude of control force, whereas the NFTSM achieves the most efficient fuel consumption.}

\begin{figure}[t]
\begin{center}
\includegraphics[width=1.0\linewidth]{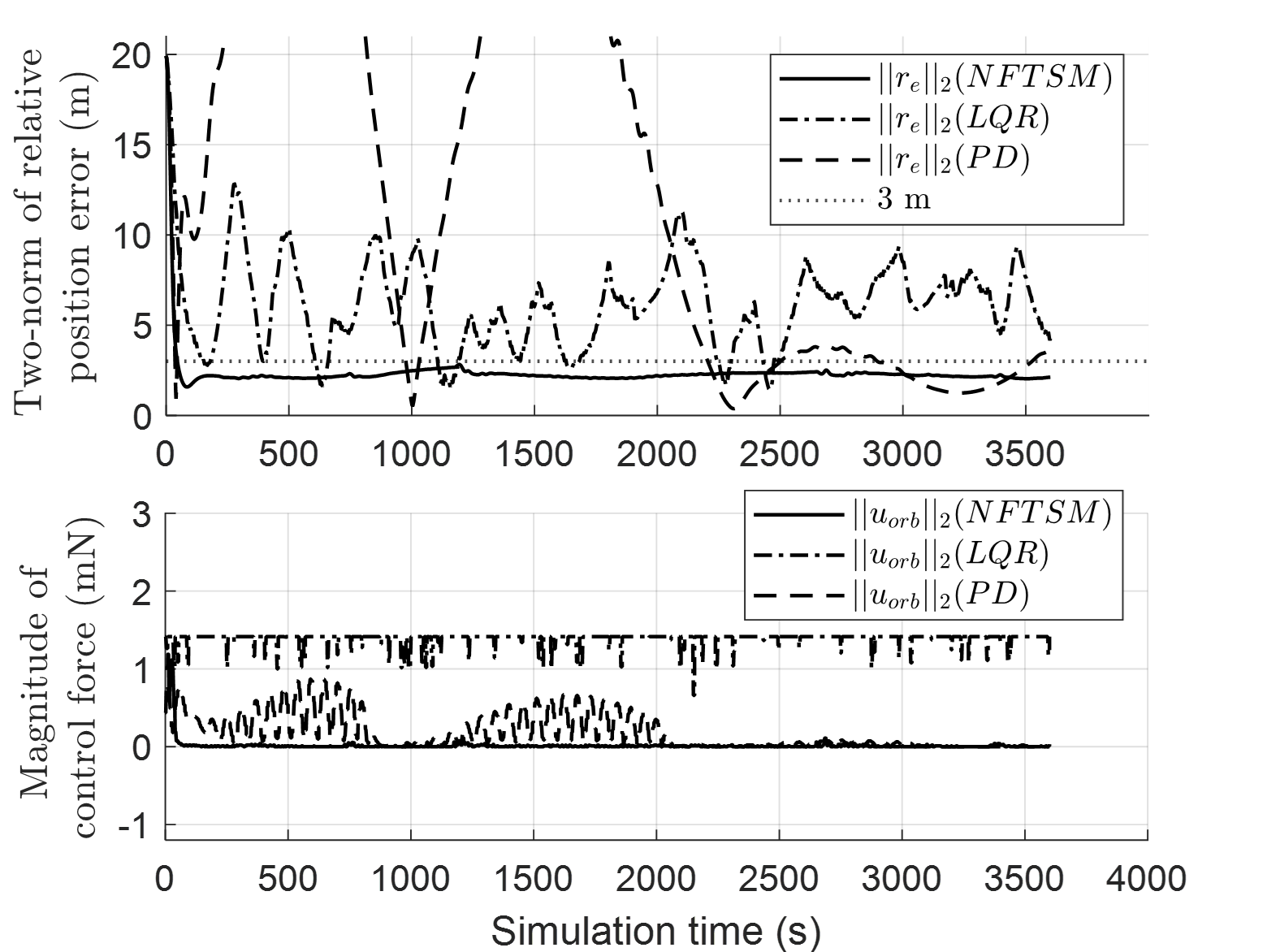}
\end{center}
\caption{Control performance comparison}
\label{fig:simAlignComp}
\end{figure}

\section{Conclusion}
\label{s:6conclusion}
The virtual telescope demonstration mission requires precise control performances in both attitude and relative orbit control systems. Since the propulsion system has only two degrees of freedom, the attitude control performance could degrade the mission success possibility. In this paper, the nonsingular fast terminal sliding surface and the adaptive smooth controller are introduced to the attitude and relative orbit control systems. The nonsingular fast terminal sliding mode naturally has a fast convergence rate and improves the alignment time. The control law is designed to generate continuous signals and prevents the system from chattering--the radical problem of the classical sliding mode control. The adaptive rule designed based on the vector norm of the sliding variable is able to consider the changes in the uncertainties and disturbances. The simulation results show that the adaptive smooth controller based on the nonsingular fast terminal sliding mode greatly improves the control performance with reasonable thrust consumption.

\section*{ACKNOWLEDGMENT}
This research was supported by the Challengeable Future Defense Technology Research and Development Program through the Agency For Defense Development(ADD) funded by the Defense Acquisition Program Administration(DAPA) in 2023(No.915027201)




\begin{IEEEbiography}[{\includegraphics[width=1in, height=1.25in,clip, keepaspectratio]{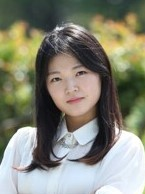}}]{Soobin Jeon}{\space}Soobin Jeon is a PhD candidate who will graduate in Aug 2024 at Yonsei University in South Korea. She received a B.S. degree in astronomy from the same university (2017). Her research field includes Guidance, Navigation, and Control (GN\&C), robust control, and satellite constellation. She was involved in CubeSat missions such as CANYVAL-C and MIMAN. This paper addresses another research interest, the robust control system development. She recently participated in the satellite constellation project that dealt with the analytical orbit design of communication constellations with continuous coverage.  
\end{IEEEbiography}

\begin{IEEEbiography}[{\includegraphics[width=1in, height=1.25in,clip, keepaspectratio]{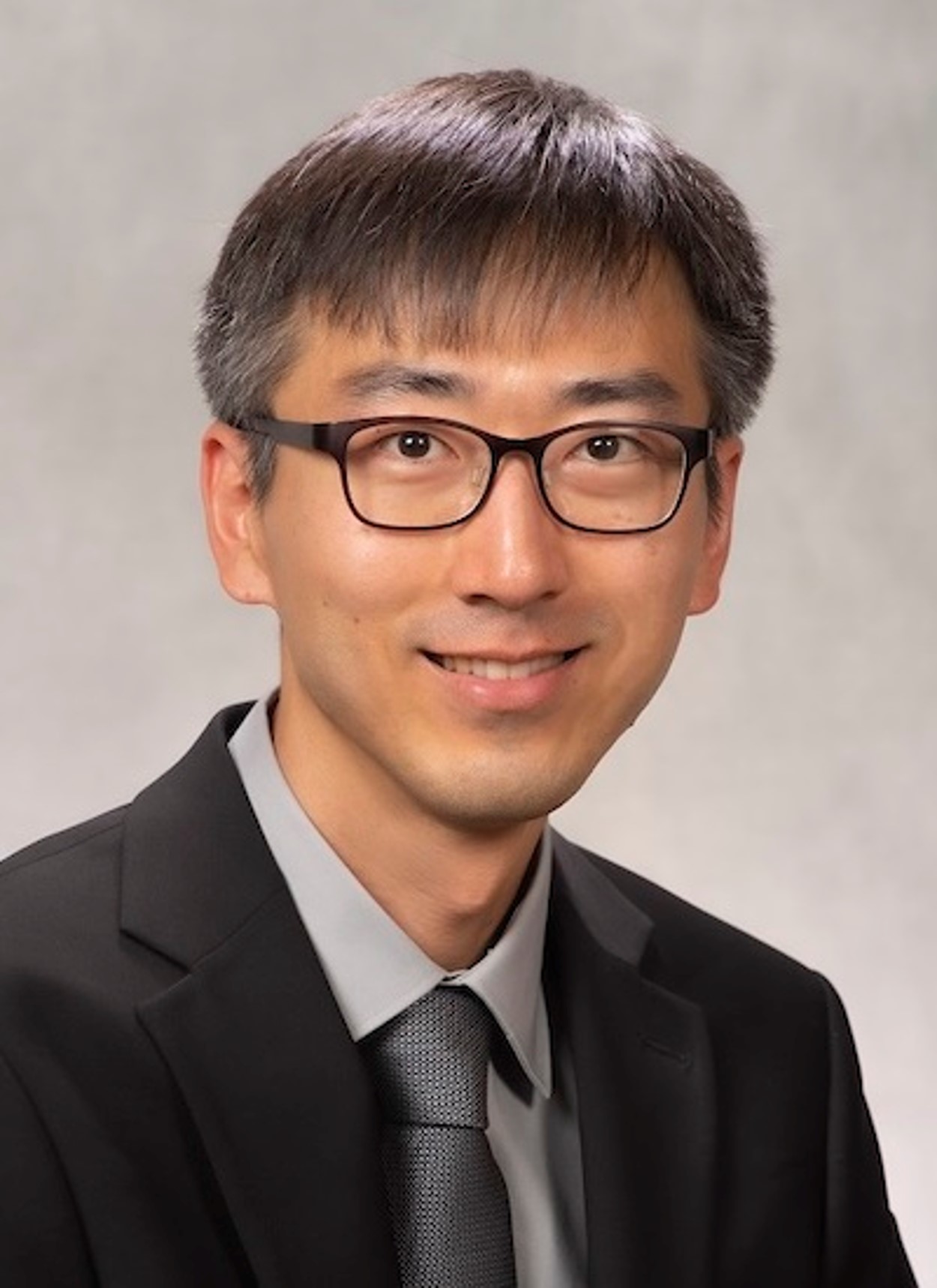}}]{Hancheol Cho}{\space}Dr. Hancheol Cho received the B.S. and M.S. degrees in astronomy from Yonsei University, South Korea, in 2006 and 2008, and the Ph.D. degree in aerospace engineering from the University of Southern California, Los Angeles, USA, in 2012.
    From 2013 to 2015, he was a senior research engineer with Samsung Techwin Co., Ltd. in South Korea. He was a Marie-Curie COFUND postdoctoral fellow with the University of Liege, Belgium and a postdoctoral appointee with Sandia National Laboratories, Albuquerque, USA. From 2021 to 2023, he was an assistant professor with the Department of Aerospace Engineering, Embry-Riddle Aeronautical University, Daytona Beach, USA. Since September 2023, he has been an associate professor with the Department of Astronomy and the Department of Satellite Systems, Yonsei University, South Korea. He also holds a visiting assistant professorship at Embry-Riddle Aeronautical University in Daytona Beach, USA. His research interests include spacecraft dynamics and control, robust adaptive controls, optimization, and robotics. He is an associate editor of the Journal of Control, Automation and Electrical Systems (Springer) and a member of the Astrodynamics Technical Committee of AIAA (American Institute of Aeronautics and Astronautics).
\end{IEEEbiography}

\begin{IEEEbiography}[{\includegraphics[width=1in, height=1.25in,clip, keepaspectratio]{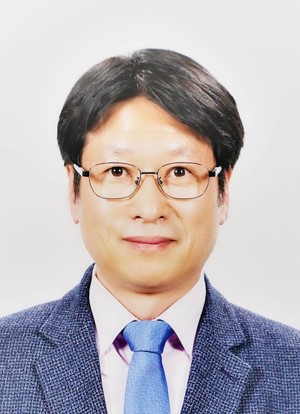}}]{Sang-Young Park}{\space}Dr. Sang-Young Park is a professor of the Astronomy Department as well as the Department of Satellite Systems at Yonsei University. He earned a Ph.D. degree in aerospace engineering from Texas A\&M University in 1996. Since joining Yonsei in 2003, he has contributed to a variety of research fields in spacecraft dynamics and control, optimal control, orbit determination, spacecraft formation flying, and CubeSat systems.

\end{IEEEbiography}

\end{document}